\newcounter{section-preserve}
\newcounter{theorem-preserve}
\newcommand{\blank}[1]{}
\newtoks\magicAppendix
\magicAppendix={}
\newtoks\magictoks
\newif\iflater
\laterfalse
\long\def\later#1{\magictoks={#1}%
	\edef\magictodo{\noexpand\magicAppendix={\the\magicAppendix 
			\the\magictoks%
	}}
	\magictodo}
\long\def\both#1{\magictoks={#1}%
	\edef\magictodo{\noexpand\magicAppendix={\the\magicAppendix 
			\noexpand\setcounter{theorem-preserve}{\noexpand\arabic{theorem}}%
			\noexpand\setcounter{theorem}{\arabic{theorem}}%
			\noexpand\setcounter{section-preserve}{\noexpand\arabic{section}}%
			\noexpand\setcounter{section}{\arabic{section}}%
			\noexpand\let\noexpand\oldsection=\noexpand\thesection
			\noexpand\def\noexpand\thesection{\thesection}
			\noexpand\let\noexpand\oldlabel=\noexpand\label
			\noexpand\let\noexpand\label=\noexpand\blank
			\the\magictoks%
			\noexpand\setcounter{theorem}{\noexpand\arabic{theorem-preserve}}%
			\noexpand\setcounter{section}{\noexpand\arabic{section-preserve}}%
			\noexpand\let\noexpand\thesection=\noexpand\oldsection
			\noexpand\let\noexpand\label=\noexpand\oldlabel
	}}
	\magictodo
	\the\magictoks}
\def\magicappendix{\latertrue \the\magicAppendix}

\newif\ifabstract
\abstractfalse
\newif\iffull
\ifabstract \fullfalse \else \fulltrue \fi

\iffull
 \long\def\both#1{#1}
 \let\later=\both
 \def\magicappendix{}
\fi
\ifabstract
\documentclass{cccg25}
\fi
\iffull
\documentclass{article}
\fi

\usepackage{graphicx,amssymb,amsmath}
\usepackage{amsfonts}
\usepackage{url, hyperref}
\usepackage{subcaption}
\usepackage{xcolor}
\usepackage{thmtools}
\usepackage{thm-restate}
\usepackage[capitalize]{cleveref}

\usepackage{algorithm}
\usepackage[noend]{algorithmic}

\renewcommand{\em}{\bf}
\renewcommand{\emph}[1]{\textbf{#1}}

\newcommand\module{\mathbf}
\newcommand{\LocateAndFree}{\operatorname{\text{\textsc{LocateAndFree}}}}

\newcommand{\Tunnel}{\operatorname{\textsc{Tunnel}}}
\newcommand{\Fix}{\operatorname{\textsc{Fix}}}

\crefname{lemma}{lemma}{lemmas}

\iffull
\newtheorem{theorem}{Theorem}

\newtheorem{lemma}[theorem]{Lemma}

\def\QED{\ensuremath{{\square}}}
\def\markatright#1{\leavevmode\unskip\nobreak\quad\hspace*{\fill}{#1}}
\newenvironment{proof}
  {\begin{trivlist}\item[\hskip\labelsep{\bf Proof.}]}
  {\markatright{\QED}\end{trivlist}}

\fi

\title{Input-Sensitive Reconfiguration of Sliding Cubes}

\author{Hugo A. Akitaya\thanks{University of Massachusetts Lowell, Lowell, USA, \texttt{\{hugo\_akitaya@, frederick\_stock@student.\}uml.edu}. Supported by NSF grant CCF-2348067.}
	\and
	Matias~Korman
  \thanks{{Siemens Electronic Design Automation, Wilsonville, USA}, 
          \texttt{matias.korman@siemens.com}}
  \and 
  Frederick~Stock%
  \footnotemark[1]
  }

\index{Akitaya, Hugo}
\index{Korman, Matias}
\index{Stock, Frederick}

\begin{document}
\thispagestyle{empty}
\maketitle

\begin{abstract}
A configuration of $n$ unit-cube-shaped \textit{modules} (or \textit{robots}) is a lattice-aligned placement of the $n$ modules so that their union is face-connected. 
The reconfiguration problem aims at finding a sequence of moves that reconfigures the modules from one given configuration to another.
The sliding cube model (in which modules are allowed to slide over the face or edge of neighboring modules) is one of the most studied theoretical models for modular robots.

In the sliding cubes model we can reconfigure between any two shapes in $O(n^2)$ moves ([Abel \textit{et al.} SoCG 2024]). If we are interested in a reconfiguration algorithm into a \textit{compact configuration}, the number of moves can be reduced to the sum of coordinates of the input configuration (a number that ranges from $\Omega(n^{4/3})$ to $O(n^2)$, [Kostitsyna \textit{et al.} SWAT 2024]). We introduce a new algorithm that combines both universal reconfiguration and an input-sensitive bound on the sum of coordinates of both configurations, with additional advantages, such as $O(1)$ amortized computation per move.

\end{abstract}
\section{Introduction}
\emph{Modular self-reconfigurable robotic systems} offer several advantages over typical robotic systems.
Such systems are composed of a set of robotic units (called \emph{modules}) that can communicate, attach, detach, and move relative to each other.
This means modules can move to change the overall system's shape, providing versatility to perform unforeseen tasks. 
While advantageous, the increased degree of freedom also incurs an algorithmic challenge, as it is not easy to determine how to use the module's operations so that the union of all modules creates a specific shape (this is known as the {\em reconfiguration} problem). 
The algorithmic community has investigated the reconfiguration problem for many variations of modular robots, be it the type of moves allowed (sliding~\cite{a.akitaya_et_al:LIPIcs.SWAT.2022.4,DP,dumitrescu2004motion}, pivoting~\cite{akitaya2021universal,sung2015reconfiguration}, and so-called square ``atoms''~\cite{aloupis2009linear,aloupis2008reconfiguration}) or shape  (such as hexagons~\cite{a.akitaya_et_al:LIPIcs.SoCG.2021.10}).

We investigate one of the most established models, which is called the \emph{sliding (hyper-)cube model}: modules are lattice-aligned unit (hyper-)cubes where two modules are \emph{adjacent} if they share a face (i.e., a $(d-1)$-dimensional facet).
In this model, two types of moves are allowed: a module can slide along the co-(hyper-)planar faces of two other adjacent modules, or it can slide along two orthogonal faces of another module (see Fig.~\ref{fig:Cubemoves}).

\begin{figure}[t]
    \centering
    \begin{subfigure}[t]{.45\linewidth}
        \centering
        \includegraphics[width=\linewidth]{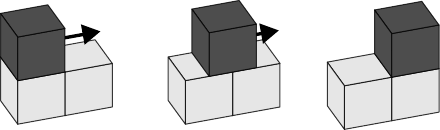}
        \caption{}
    \end{subfigure}\hfill
    \begin{subfigure}[t]{.45\linewidth}
        \centering
        \includegraphics[width=\linewidth]{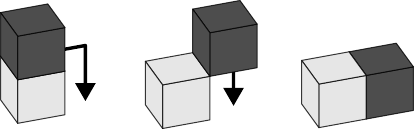}
         \caption{}
    \end{subfigure}
    \caption{
    (a) Slide and (b) Convex transition.\vspace{-1em}}
    \label{fig:Cubemoves}
\end{figure}

The \emph{reachability} problem asks whether, given two configurations with $n$ modules, there exists a sequence of moves that can transform one into the other. If the answer is always `yes' we say that the model admits \emph{universal} reconfiguration.

The first universal result for the two dimensional sliding model was obtained by Dumitrescu and Pach~\cite{DP} with a reconfiguration algorithm that required $O(n^2)$ moves. This result could be seen as optimal since some instances require $\Omega(n^2)$ moves.
The 2-dimensional result was afterwards extended to three and higher dimensions in an informal publication of  Abel and Kominers~\cite{abel2008universal}. Their paper claimed universal reconfiguration and that $O(n^3)$ many moves were always sufficient.

\iffull
Other research~\cite{akitaya2021universal} focused on variations of the model, in both the shape of the modules and the type of movements that were allowed. They obtained characterization of which variations of the model have universal reconfiguration. For the versions in which no universal reconfiguration algorithm exists, they showed that determining if a particular instance can be reconfigured into a target configuration is PSPACE-complete.
\fi

Recently there have been two simultaneous, independent results in three dimensions. Abel \textit{et al.}~\cite{aakks-quadratic04} formally published their preliminary result, reducing the number of moves to $O(n^2)$. In the same paper, the authors explain how the algorithm can be made in-place and input-sensitive (the exact bound depends on the dimension of the bounding box of both configurations). \emph{In-place} means throughout the reconfiguration there are never more than a constant number of modules outside the bounding boxes of the start and end configuration. Further, these modules are never more than distance 1 from the bounding box \footnote{Traditionally, in-place algorithms are a bit more restrictive, only one module is allowed outside the bounding box, making both constants exactly 1. In this paper we relax the condition for ease of description.}.
Parallel to ~\cite{aakks-quadratic04}, Kostitsyna \textit{et al.}~\cite{inPlaceCompact24} independently produced a different input-sensitive reconfiguration algorithm: a configuration $C$ can be reconfigured into a {\em compact configuration}, intuitively, a shape where all modules are clumped together without holes (formal definition given in Section~\ref{sec:def}). Their algorithm takes $O(\sum_{\module m\in C_1} \|\module m\|_1)$ sliding moves where $\|\module m\|_1$ denotes the $L_1$ norm of the position of a module $\module m$.
This result has a better input-sensitive bound, but is slightly limited: they only show how to reconfigure into a compact configuration. Since there is more than one compact configuration with the same number of modules, this does not directly lead to a universal reconfiguration strategy. Reconfiguration between compact configurations is significantly simpler than compactification, but does not immediately follow.

This paper we improve on both ~\cite{aakks-quadratic04} and~\cite{inPlaceCompact24}: we give a universal reconfiguration whose number total number of moves depends on the sum of input and target coordinates ($O(\sum_{\module m\in C_1} \|\module m\|_1+ \sum_{\module m\in C_2} \|\module m\|_1)$. Since configurations are connected, it is easy to see that this bound is in $\Omega(n^{4/3}) \cap O(n^2)$ (when all modules form a solid cube or when they form the 1-skeleton of a larger cube, respectively). Our main result is formally stated as follows.

\begin{theorem}
    \label{thm:main}
    Given two configurations $C_1$ and $C_2$ with $n$ cube modules in 3 dimensions where all modules lie in the positive $xyz$ orthant, there is an in-place reconfiguration between them that uses at most $18(\sum_{\module m\in C_1} \|\module m\|_1+\sum_{\module m\in C_2} \|\module m\|_1) + 120n + O(1)$ sliding moves where $\|\module m\|_1$ denotes the $L_1$ norm of the position of a module $\module m$. Such a reconfiguration can be computed in $O(1)$ amortized time per sliding move.
\end{theorem}

In our algorithm we use a constant number of auxiliary ``helper'' modules that will allow nearby modules to move. We follow existing literature ~\cite{akitaya2021universal} and refer to the auxiliary modules as {\em musketeers}. The multiplicative factor 
$120$ may seem large, we note that  our algorithm uses six musketeer modules, so this constant is better characterized by $20n$ moves per musketeer module. Further, since we reconfigure between $C_1$ and $C_2$ by reaching an intermediate configuration, this means that we actually need $10n$ moves per musketeer to transform any configuration of $n$ modules into \textit{compact} form.

Although the main emphasis is in the total number of moves required, it is also interesting to bound the time needed to compute the sequence of moves needed to reconfigure between two configurations. Our algorithm can compute each move in $O(1)$ amortized time. Although the computation time is not directly addressed in the previous papers, it is not hard to see that the algorithm of~\cite{aakks-quadratic04} also achieves $O(1)$ amortized (after a BFS traversal on both configurations, each step is easily computed in O(1) time). On the other hand, a na\"ive implementation of~\cite{inPlaceCompact24} would require $\Theta(n)$ time per move as a global search must be executed each time.

\section{Definitions and Preliminaries}\label{sec:def}

A \emph{configuration} $C$ of $n$ unit cube modules is a set of $n$ cells in the cube lattice. 
A cell is \emph{occupied} if it is in $C$ or \emph{empty} otherwise.
We abuse notation by sometimes referring to occupied cells as modules.
Two cells are \emph{adjacent} if they share a face. Let the \emph{adjacency graph} $G_C$ be the graph whose vertices are the cells in $C$ and edges are defined by pairs of adjacent occupied cells. 
We call $C$ \emph{connected} if $G_C$ is connected. A module is \emph{articulate} if it is a cut vertex in $G_C$, and \emph{nonarticulate} otherwise.
For an occupied cell $m$, the notation $\module m$ refers to the module at cell $m$. 
Note that a connected configuration $C$ defines a polycube. We denote by $\partial C$ the boundary of $C$. The \emph{outer boundary} of $C$ is the boundary of the unbounded component of the complement $\overline{C}$ of $C$.
We say a module is ``in the outer boundary'' if at least one of its faces lies in the outer boundary of $C$.


A \emph{move} is an operation that transforms an $n$-module configuration $C$ into another $C'$ so that $C\cap C'$ is a configuration of $n-1$ modules (and thus $C \setminus C'$ is a set with one module).
We require the \emph{single backbone condition}: a move between connected configurations $C$ and $C'$ is only allowed if $C\cap C'$ is also connected.
We say that the module in position $C\setminus C'$ \emph{moved} to position $C'\setminus C$. In other words, the single-backbone condition requires each moving module to be nonarticulate. 

The \emph{sliding model} allows two types of moves (refer to Fig.~\ref{fig:Cubemoves}):
\begin{itemize}
  \item A \emph{slide} moves a module $\module a$ from $a$ to an adjacent empty cell $b$, and requires that there are adjacent occupied cells $a'$ and $b'$ such that $a$ is adjacent to $a'$ and $b$ is adjacent to $b'$.
  \item A \emph{convex transition} moves a module $\module a$ from $a$ to an empty cell $b$ where $a$ and $b$ share a common edge $e$, and are both adjacent to an occupied cell $c$, and requires that the cell $d\notin\{a,b,c\}$ that contains $e$ is empty. Note that every edge $e$ is incident to exactly $4$ cells.
\end{itemize}

Note a slide requires the final cell $b$ to be empty, and a convex transition requires the target $b$ as well as an intermediate cell to be empty.
These are called the \emph{free-space requirements} of the moves. If these requirements are not met, performing either of these moves would cause two modules to collide.
Note that a module might be nonarticulate but is not allowed to move if the free-space requirements mentioned above are not satisfied.
We call a module \emph{movable} if it is nonarticulate and satisfies the free-space requirements of either a slide or a convex transition.

We may refer to a cell using the coordinates of its closest corner to the origin. 
Hence cell $(0,0,0)$ corresponds to the unit cube with vertices $(0,0,0)$, $(0,0,1)$, $(0,1,0)$, $(1,0,0)$, $(0,1,1)$, $(1,1,0)$ and $(1,1,1)$.
We denote by $\|\module m\|_1$ the $L_1$ norm of the cell occupied by a module $\module m$.
For example, if $\module m$ occupies $(x_{\module m}, y_{\module m}, z_{\module m})$, then $\|\module m\|_1=|x_{\module m}| + |y_{\module m}| + |z_{\module m}|$. 

We assume that both configurations are contained in the positive orthant and that a corner of the bounding box of either configuration is the origin.

Our algorithm uses the {\em slice graph} from Fitch and Rus~\cite{fitch2003reconfiguration}. For a connected configuration $C$ of $n$ modules and $z_0 \in \mathbb{Z}$, the {\em slice} at height $z$ is $C \cap \{z=z_0\}$ (the modules whose $z$-coordinate is equal to $z_0$). Each maximally connected component in a slice is called a {\em cluster}. Each cluster defines a vertex of the slice graph. Two vertices are connected if at least one module in each cluster share a face (see Fig.~\ref{fig:slice-graph}).

\begin{figure}[ht!]
    \centering
    \includegraphics[scale=0.5]{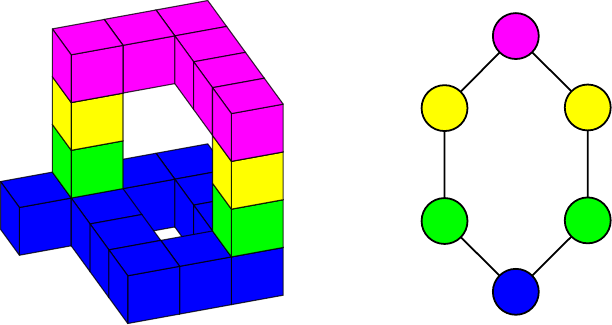}
    \caption{A configuration (left) and its slice graph (right). }
    \label{fig:slice-graph}
\end{figure} 

A cell $c$ \emph{dominates} another cell $c'$ if $c'$ is contained in the minimum axis aligned bounding box containing $c$ and the origin. We also say that $c=(x,y,z)$ is \emph{below} $c'=(x',y',z')$ if $x=x'$, $y=y'$ and $z< z'$ (similarly, we define \emph{above} in a similar way.

A module is called \emph{compact} if every cell dominated by it is occupied.
A module $\module m$ is called \emph{quasi-compact} if every cell that is dominated by $\module m$ and is not below is occupied. A configuration or cluster is \emph{compact} (resp., \emph{quasi-compact}) if every module in it is compact (resp., quasi-compact).
A cluster is \emph{extremal} if not all of its modules are quasi-compact, and all adjacent modules above it (if any) are quasi-compact. (Note that, by definition, a cluster above a non-quasi-compact cluster may be quasi-compact.)

\subsection{General strategy.} A common technique in reconfiguration algorithms is to define a \emph{canonical} configuration $C^*$ and describe a sequence of moves that transform a given configuration $C$ into $C^*$. 
Because the moves are reversible, this implies a solution for the reachability problem: Given two configurations $C$ and $C'$, one can obtain reconfiguration sequences from both to $C^*$, and then compose the reconfiguration from $C$ to $C^*$ with the reverse of the reconfiguration from $C'$ to $C^*$.
Instead of a single configuration, we use a class of canonical configurations: all compact configurations of $n$ modules.
The strategy is to compute a sequence of moves to transform $C$ and $C'$ to two configurations $D$ and $D'$ of this class. We then transform $C$ into $D$, then $D$ into $D'$, and finally reconfigure $D'$ to $C'$. 
Thus, our algorithm can be divided into two phases: reconfiguration between a given configuration and a compact one (Section~\ref{sec:compactify}); and reconfiguration between two compact configurations (Section~\ref{sec:comp2comp}). 

\section{Compacting configurations}\label{sec:compactify}
In this section, we describe an algorithm we call \textsc{Compactify}, which reconfigures a given configuration into a compact one.
As a preprocessing step, use $\LocateAndFree$ (from \cite{aakks-quadratic04}) to obtain six movable modules on the outer boundary. These modules will be referred to as {\em musketeers}.

Once done, the algorithm repeatedly picks an extremal cluster $S$ and uses the musketeers to move every module in $S$ without a lower neighbor to a position with a smaller $z$-coordinate. 
By doing so, $S$ will merge with at least one cluster in the slice below. 
Initially, a module is only moved if the position below it is empty. Otherwise, it is left where it was. We call these leftover modules \emph{stragglers}. We use another procedure $\Fix$, to find a suitable position for the stragglers.
The result is that every module previously in $S$ either has a lower $z$-coordinate or becomes quasi-compact. By iterating over all clusters we will finish with a compact configuration or all modules in the $z=0$ plane.
Pseudocode of $\textsc{Compactify}$ is presented in Alg.~\ref{alg:Comp} . 

\both{
\subsection{Obtaining Musketeers}\label{sec:ob-musket}}
The main role of the musketeers will be to provide temporary connectivity to nearby modules. We use a subroutine $\LocateAndFree$ from~\cite{aakks-quadratic04} to gather them above the highest cluster.

\both{
\begin{lemma}
\label{lem:buildMusketeers}
A constant number $k$ of musketeer modules can be positioned on the outer boundary of any configuration in at most $4kn+O(1)$ moves.
\end{lemma}
}
\later{
\begin{proof}
The proof follows from~\cite[Lemma 6]{aakks-quadratic04}, which states that we can reconfigure a configuration $C$ such that its outer face stays the same and one of its modules in the outer face becomes nonarticulate. 
Due to this, the reconfiguration of $C \setminus \{\module m_1\}$ does not require coordination with ${\module m_1}$.
If ${\module m_2}$, becomes nonarticulate after the second application of $\LocateAndFree$, it happens to be the only neighbor of ${\module m_1}$, we exchange their labels.
We call $\LocateAndFree$ a constant number of times and each call requires $O(n)$ moves~\cite[Lemma 8]{aakks-quadratic04}. 
Although not explicitly stated in~\cite{aakks-quadratic04}, each call of $\LocateAndFree$ takes at most $4n+O(1)$ moves because each module that moves during its execution visits a disjoint set of faces of the static modules, and each face is visited at most once.
The maximum surface area of a polycube with $n$ cubes is $4n+O(1)$: inductively, each added cube removes at least one face and adds at most 5 faces to the boundary.
\end{proof}
}


\later{
\begin{algorithm}[!ht]

  \caption{$\textsc{Compactify}$(C)}
  \label{alg:Comp}
  \begin{algorithmic}[1]
    \STATE Compute the slice graph of $C$.
    \WHILE{$C$ is not quasi-compact (ignoring slice $z=0$)}  \label{ln:comp-while}
    \STATE Let $R$ be the cluster where the musketeers currently sit. DFS on the slice graph from $R$, preferring upward edges (in the $z$ direction). Let $S$ be the first extremal cluster (with possibly $S=R$).
    \STATE Move the musketeers to $S$.\label{ln:muskToCluster}
    \STATE Lower non-quasi-compact modules in $S$ as in \cref{sec:ap-lower}.\label{ln:lower}
    \STATE Let $M$ be the set of stragglers $\module m$ sorted by $\|\module m\|_1$.\label{ln:straggler-order}
    \FOR{$\module m \in M$}
        \STATE \hyperref[alg:fix]{$\Fix$}($\module m$)\label{ln:cmptx-fixCall}
    \ENDFOR
    \ENDWHILE
    \WHILE{$\exists \module m$ where $\module m+(0,0,-1)$ is empty and $z_{\module m}>0$}\label{ln:compacting-quasi}
    \STATE Slide $\module m$ down\label{ln:slide-m}
    \ENDWHILE
    \IF{$C$ is not compact}
        \STATE Apply the planar version of \textsc{Compactify} on slice $z=0$
    \ENDIF
  \end{algorithmic}
\end{algorithm}
}

\begin{figure}[t]
    \centering
    \begin{subfigure}[b]{0.32\linewidth}
        \centering
        \includegraphics[width=.9\linewidth]{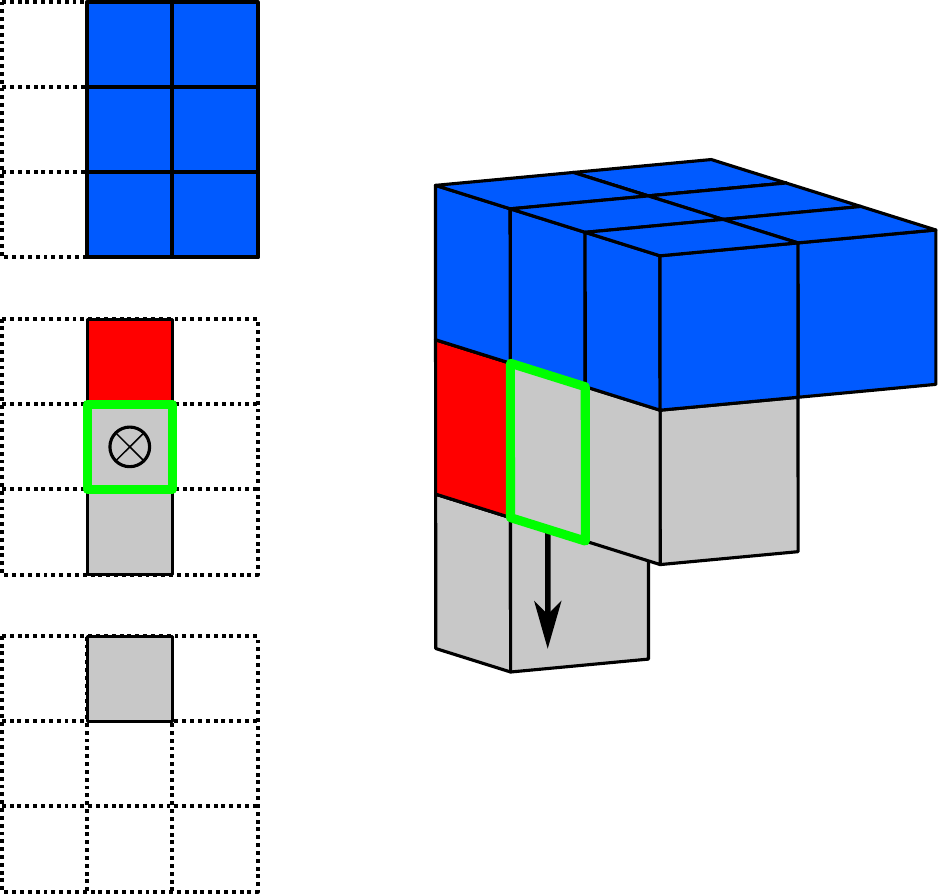}
        \caption{}
        \label{fig:lower-str-a}
    \end{subfigure}%
    \hfill
    \begin{subfigure}[b]{0.32\linewidth}
        \centering
        \includegraphics[width=.9\linewidth]{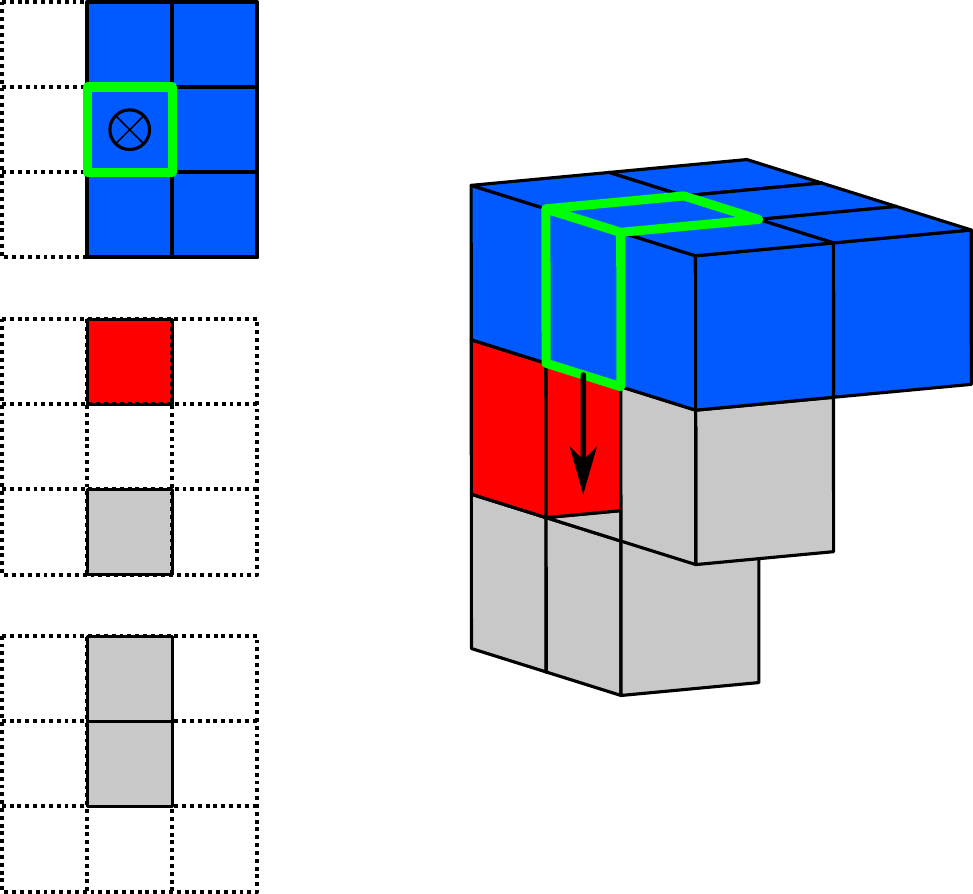}
        \caption{}
    \end{subfigure}%
    \hfill
    \begin{subfigure}[b]{0.32\linewidth}
        \centering
        \includegraphics[width=.9\linewidth]{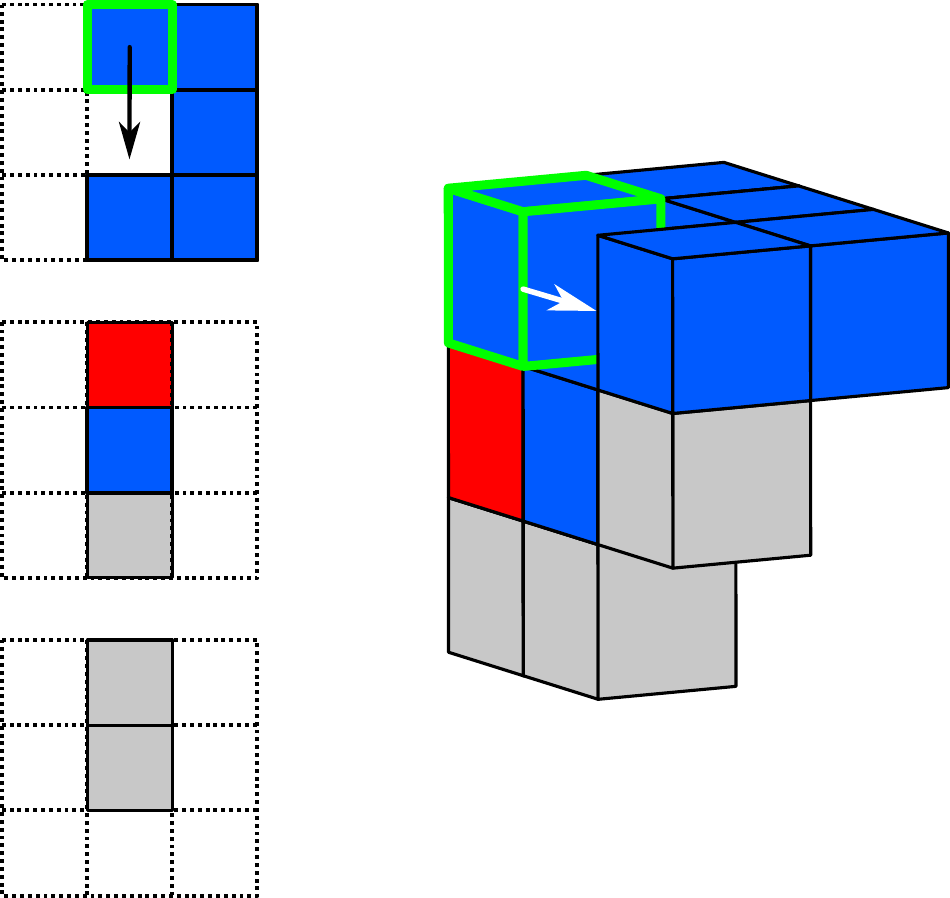}
        \caption{}
    \end{subfigure}%
   
    \begin{subfigure}[b]{0.32\linewidth}
        \centering
        \includegraphics[width=.9\linewidth]{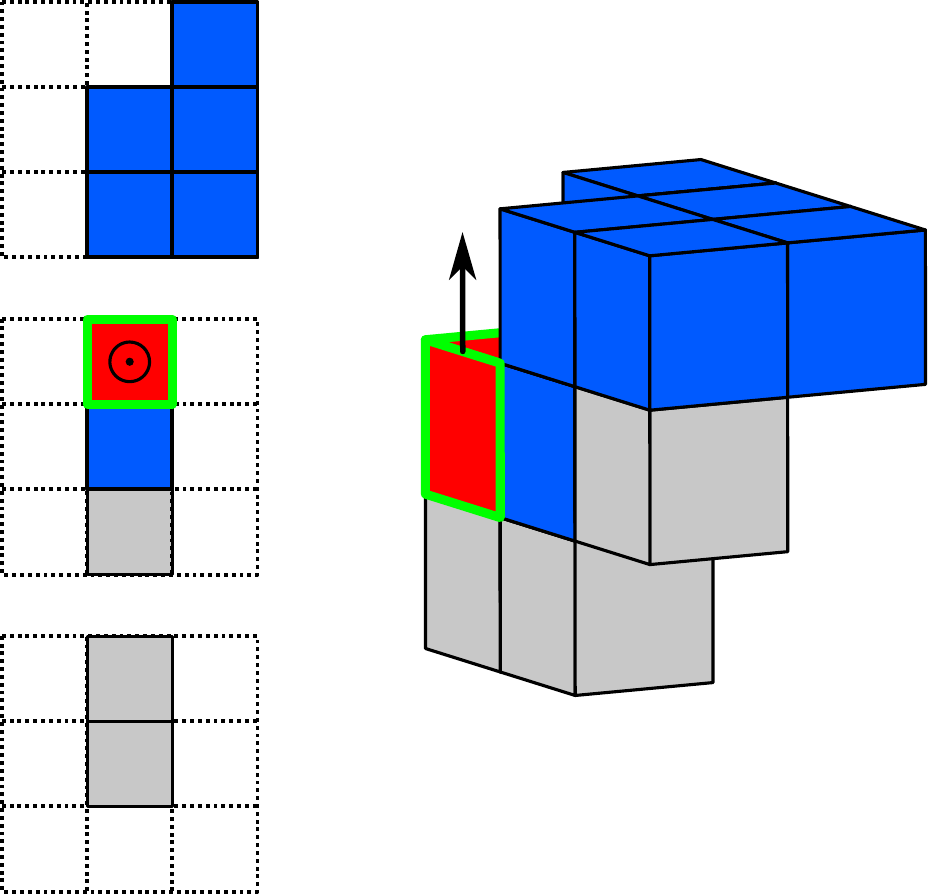}
        \caption{}
    \end{subfigure}%
    \:\:
    \begin{subfigure}[b]{0.32\linewidth}
        \centering
        \includegraphics[width=.9\linewidth,page=10]{Figures/six-general.pdf}
        \caption{}
    \end{subfigure}%
    \caption{The four moves used to lower a straight degree-2 module in an extremal cluster using the musketeers for connectivity. The left of each figure shows the top view of the three relevant $z$-layers. $\odot$ and $\otimes$ represent arrows going into and out of the plane, respectively. The moving module is highlighted in green.}
    \label{fig:lower-straight}
\end{figure}

\begin{figure}[t]
    \centering
    \begin{subfigure}[b]{0.32\linewidth}
        \centering
        \includegraphics[width=.9\linewidth]{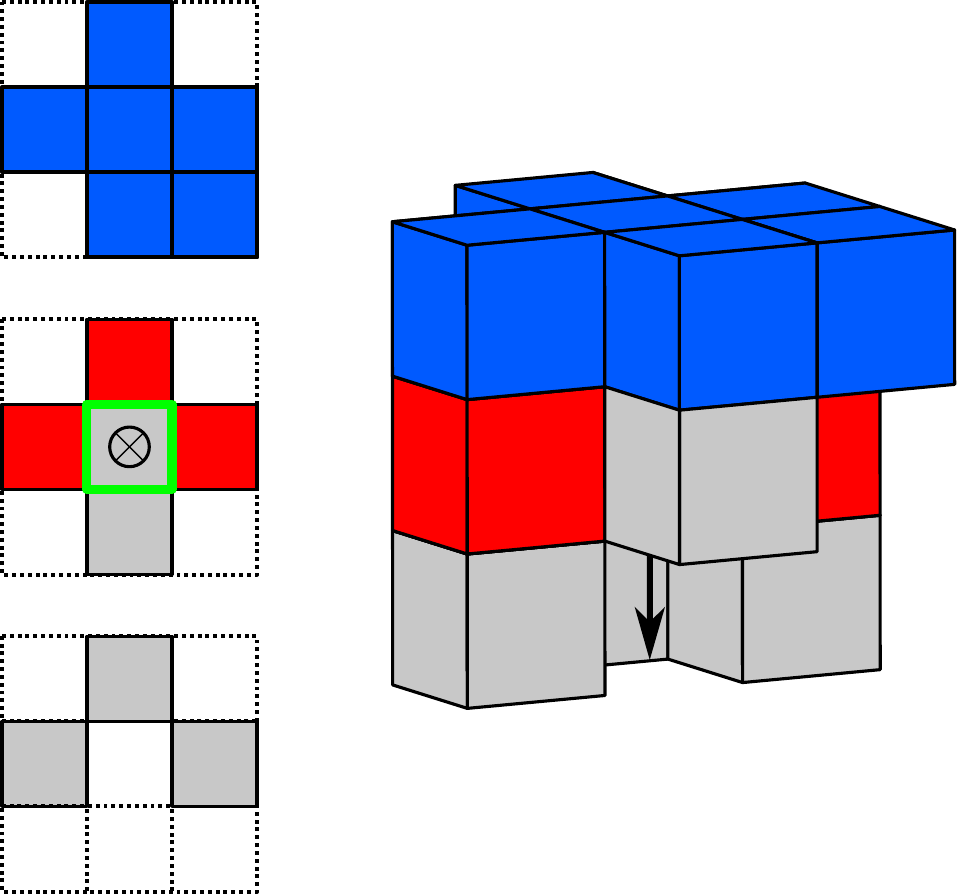}
        \caption{}
        \label{fig:lower-bend-a}
    \end{subfigure}%
    \hfill
    \begin{subfigure}[b]{0.32\linewidth}
        \centering
        \includegraphics[width=.9\linewidth]{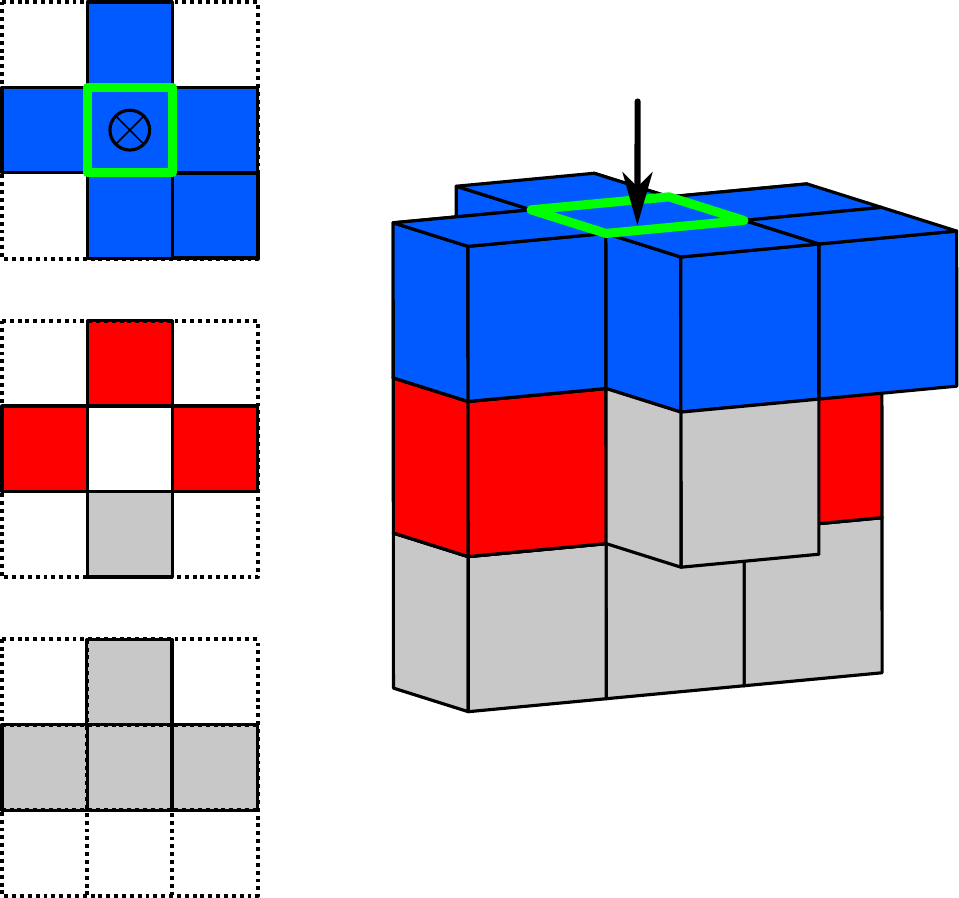}
        \caption{}
    \end{subfigure}%
    \hfill
    \begin{subfigure}[b]{0.32\linewidth}
        \centering
        \includegraphics[width=.9\linewidth]{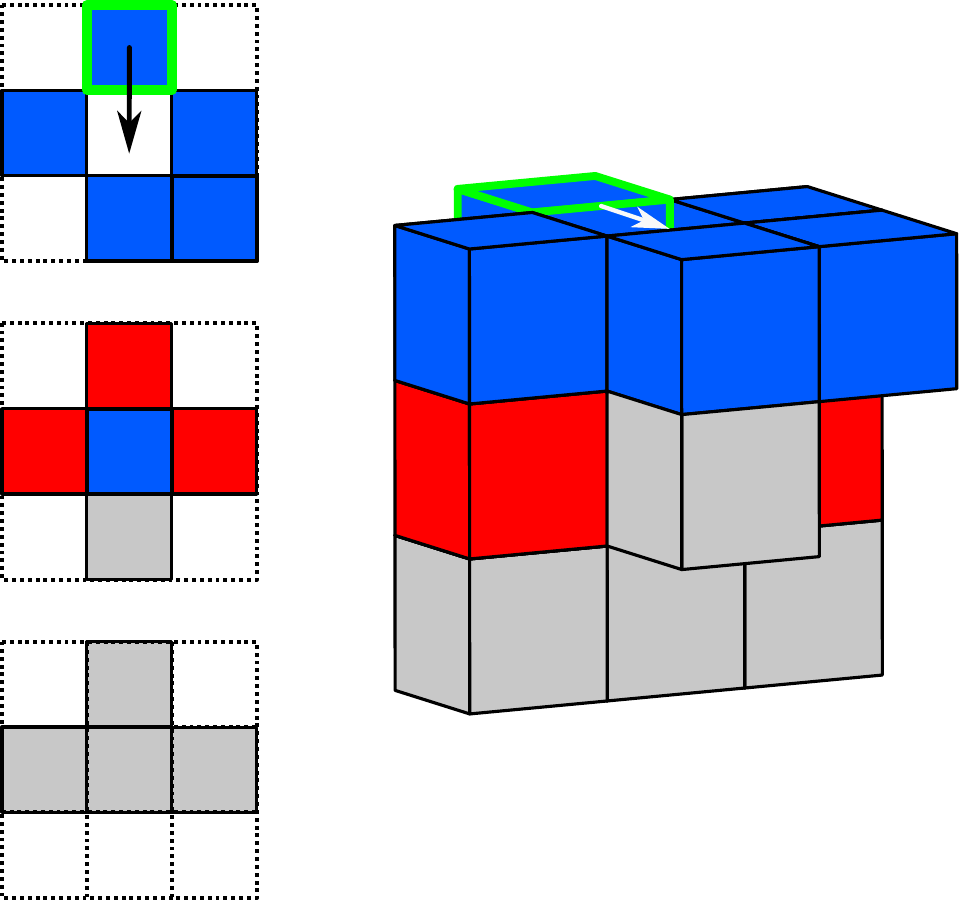}
        \caption{}
    \end{subfigure}%
    
    \begin{subfigure}[b]{0.32\linewidth}
        \centering
        \includegraphics[width=.9\linewidth]{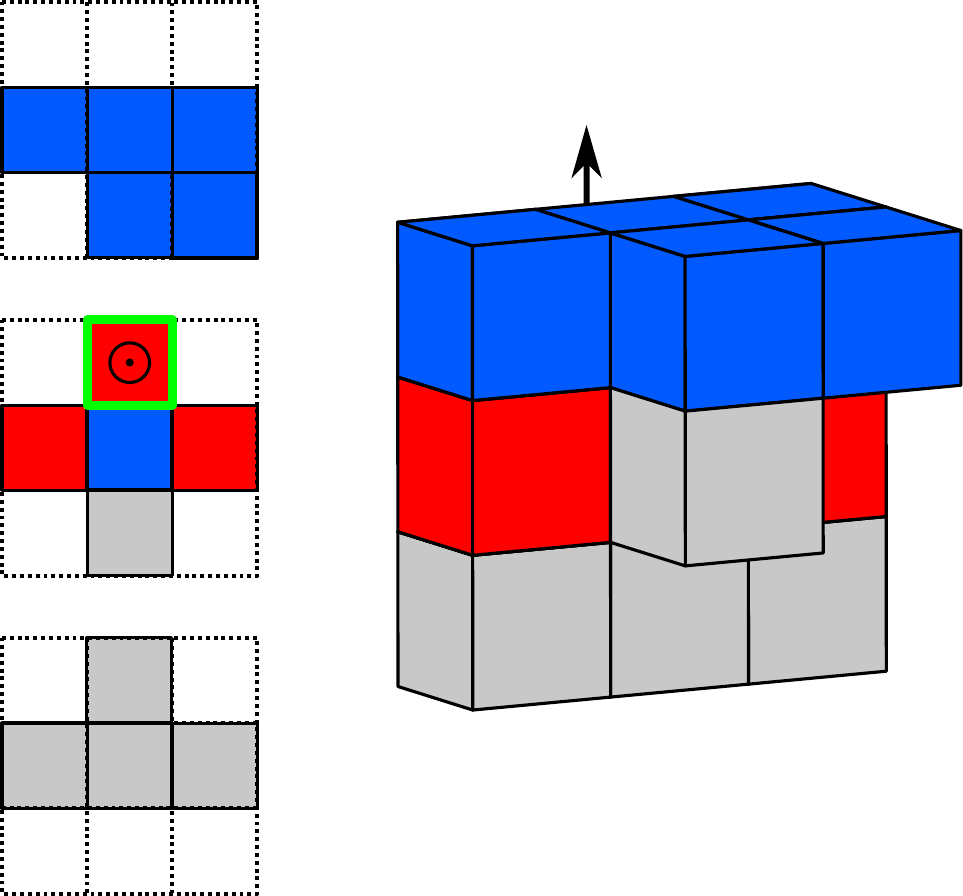}
        \caption{}
    \end{subfigure}%
    \:\:
    \begin{subfigure}[b]{0.32\linewidth}
        \centering
        \includegraphics[width=.9\linewidth,page=5]{Figures/six-general.pdf}
        \caption{}
    \end{subfigure}%
    \caption{The four moves used to lower a degree-4 module in an extremal cluster using the musketeers for connectivity. }
    \label{fig:lower-bend}
\end{figure}

\later{
\begin{figure}[t]
    \centering
    \begin{subfigure}[b]{0.25\linewidth}
        \centering        \includegraphics[scale=.2]{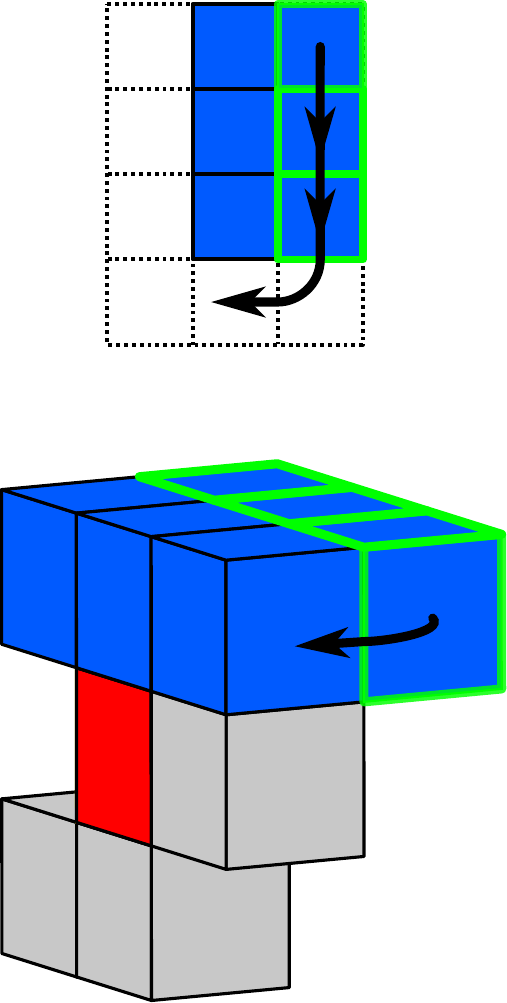}
        \caption{}
    \end{subfigure}%
    \begin{subfigure}[b]{0.25\linewidth}
        \centering        \includegraphics[scale=.2]{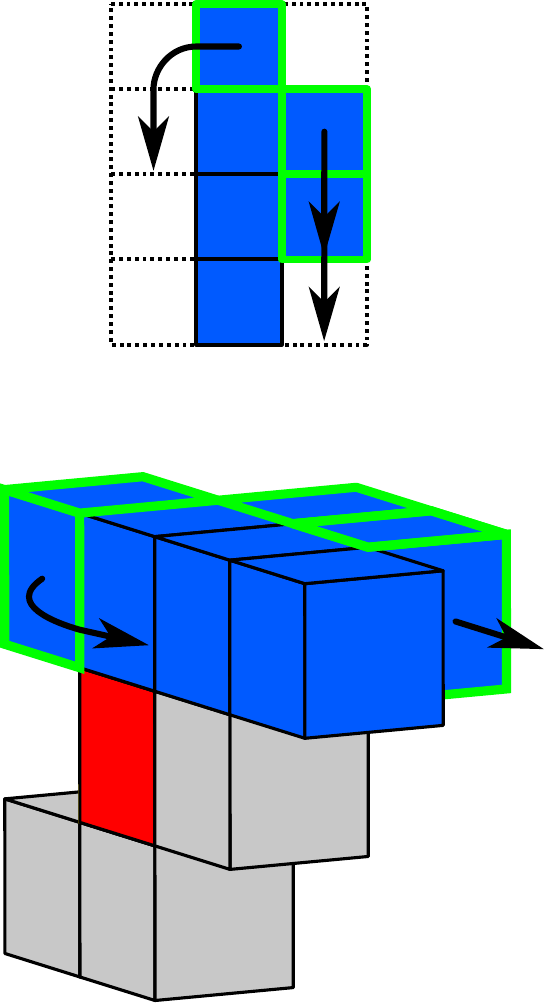}
        \caption{}
    \end{subfigure}%
    \begin{subfigure}[b]{0.25\linewidth}
        \centering        \includegraphics[scale=.2]{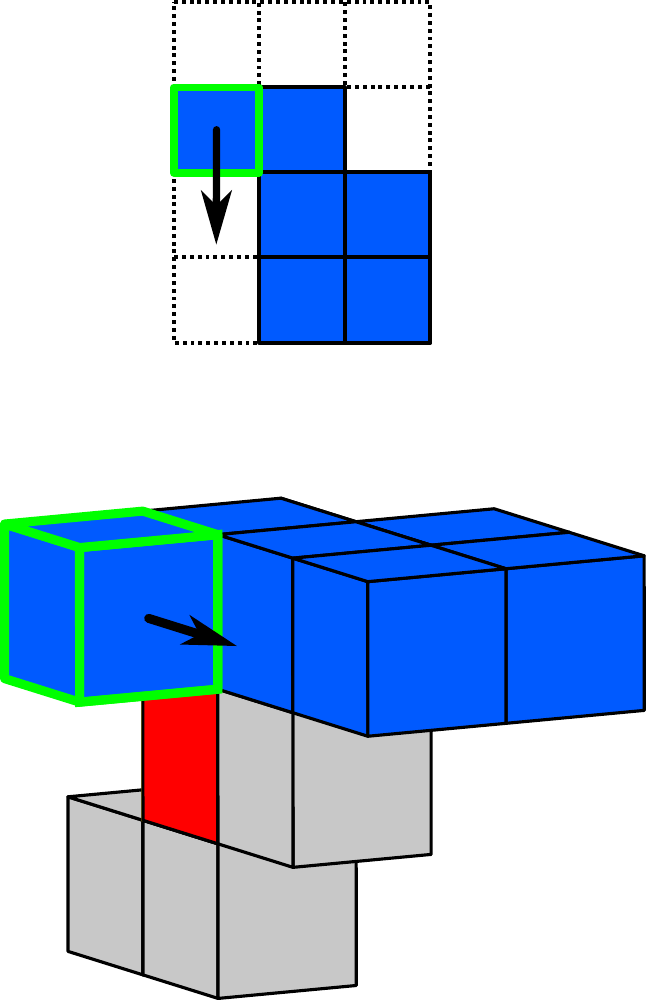}
        \caption{}
    \end{subfigure}%
    \begin{subfigure}[b]{0.25\linewidth}
        \centering        \includegraphics[scale=.2,page=14]{Figures/six-general.pdf}
        \caption{}
    \end{subfigure}%
    \caption{With 6 moves, the musketeers can advance by one in the post order above the extremal cluster.}
    \label{fig:lower-advance}
\end{figure}
}
\subsection{Lowering Modules}
\label{sec:ap-lower}

Given an extremal cluster $S$ let $S^*$ be the set of non-quasi-compact modules in $S$, and let $T$ be a minimum Steiner tree of $S^*$ using the quasi-compact modules in $S$ (if any) as Steiner vertices.
We process $T$ in post-order, rooted arbitrarily, moving it to the z-layer below if it has no lower neighbor.
Lowering is accomplished via the musketeer modules (shown in blue in \Cref{fig:lower-straight,fig:lower-bend}). The red modules are modules that have already been processed. If $\module p$ has a lower neighbor, it cannot move directly down, thus it becomes a straggler. (After completing our traversal of $T$ we return and ``fix'' them.). 

Since we visit modules in post order, we can guarantee the children of the module currently being processed (if any) have been already processed. This implies that, when processing a module, the neighbors in the same z-layer are either red or is the parent of the module in $T$. The exact strategy depends on the degree of $\module p$ in $T$.

\later{
If $\module p$ is a non-straggler degree-1, simply move it below its parent. (This does not require the presence of the musketeers to aid in connectivity.)
We maintain the invariant that every module that has been processed (and therefore lies in the $z$-layer below its original position) and whose parent have not yet been processed has an upper neighbor connecting it to its parent (shown in red in \cref{fig:lower-straight,fig:lower-bend}). 
If $\module p$ has degree 2 in $T$ where $\module p$, its child, and its parent are collinear (``straight'' degree-2), we use the formation of musketeers shown in \cref{fig:lower-straight} and, with the four moves shown in the figure, we move $\module p$ down one unit while maintaining the invariant.
Note that in the last step, the red module swaps roles with a musketeer, becoming a musketeer itself while one of the musketeers becomes a red module.
In the remaining cases ($\module p$ has degree 3 or 4), we can assume without loss of generality that $\module p$ has a child at position $(1,0,0)$ (relative to $\module p$) and its parent is at position $(0,-1,0)$.
\cref{fig:lower-bend} shows the formation of musketeers and the four moves that lower $\module p$ while maintaining the invariant if $\module p$ has degree 4 in $T$.
If $\module p$ has degree 3, note that the second move would disconnect the musketeers at $\module p+(0,1,1)$ or $\module p+(-1,0,1)$ if $\module p$ has no child at $\module p+(0,1,0)$ or $\module p+(-1,0,0)$, respectively.
In such cases we can replace this move by moving one of these two musketeers that has no lower neighbor into the original position of $\module p$.
We can then get a red module to the previous position of the last moving musketeer in two moves.
As before, the roles of a musketeer and a red module get swapped. 
The remaining red modules are not needed anymore for connectivity and thus get marked as stragglers (shown in dark gray in \cref{fig:lower-bend}(e)).
Therefore, after processing $\module p$, processing its parent requires one of two possible configurations of musketeers. It is relatively easy to verify only six moves are required to form these. \cref{fig:lower-advance} shows an example moving between different formations.
After processing every module in $T$ we move some stragglers as follows. We can assign each original leaf of $T$ to a straggler that either was created at the neighborhood of a high degree node (dark grey in \cref{fig:lower-bend}(e)) or was created at the root of $T$. (Note that the root of $T$ always becomes a straggler by construction).
We do that so that the paths between the stragglers and assigned leaves in $T$ are disjoint.
Move such stragglers along $T$ to the position immediately below the original position of their assigned leaves.
}

\both{\begin{lemma} \label{lem:melt}
    Given a Steiner tree $T$ of the non-quasi-compact modules of an extremal cluster $S$ and 6 musketeer modules positioned above $T$ as in \cref{fig:lower-str-a} or \cref{fig:lower-bend-a}, 
    let $\ell$ be the number of modules that were lowered and $s$ be the number of remaining stragglers ($|T|=\ell+s$).
    Then, $17\ell+13s$ moves suffice to move down by one unit every module in $T$ that has no lower neighbor while maintaining connectivity.  
\end{lemma}}
\later{
\begin{proof}
    Let $S' \subseteq S$ contain the modules of $S$ that have already been lowered (initially empty).
    By the invariant, modules that have been processed but their parent has not (red modules in~\ref{fig:lower-advance}) connect $S'$ and $S$.
    By construction, the musketeers placed above module $\module p$ currently being processed act like a bridge connecting the red descendants of $\module p$ with the parent of $\module p$.
    After $\module p$ is lowered, one of the red modules become a musketeer and a musketeer is placed above $\module p$.
    During this movement, the recently lowered $\module p$ connects all its descendants and the remaining musketeers connect a red module to the parent of $\module p$.
    This maintains the invariant and, thus, connectivity.

    By the nature of the post order, the musketeers may visit a module at most 4 times (as many times as the maximum degree of $T$). On the other hand, leaves of $T$ do not need the musketeers for processing, and we can charge to each module at most two visits of the musketeers. 
    That amounts to $12$ moves per module. 
    Lastly, we can match the stragglers that were red modules to leaves so that their paths are disjoint in the last phase of the lowering algorithm. Thus, each module is adjacent to these moving stragglers at most once. 
    It takes 4 moves to lower $\module p$.
    Note that these moves do not apply for stragglers. 
    The total moves per lowered module is then 17, and for stragglers is 13.   
\end{proof}
}

\both{
\subsection{Handling straggler modules}\label{sec:fix}
}
\later{
This section provides details on the handling of straggler modules (procedure $\Fix$).
Intuitively speaking, each straggler $\module m$ searches for an empty position $p$ dominated by and not directly below itself. The goal is to move $\module m$ to $p$, either directly or through a chain of moves that fill $p$ and leave $m$ empty (\Cref{fig:fix-xray}).
The pseudo-code for $\Fix$ is presented in Alg~\ref{alg:fix}. 

\begin{algorithm}[htb]
  \caption{$\textsc{Fix}(\module m)$}
  \label{alg:fix} %
  \begin{algorithmic}[1] %
    \WHILE{$\module m$ can move monotonically in the direction $-\module m$}\label{ln:fix-first-line} \label{ln:fix-1}
        \STATE{Move $\module m$ closer to the origin} \label{ln:fix-greedy}
        \IF{$\module m$ moved to a lower $z$-coordinate or became quasi-compact}\RETURN
        \ENDIF
    \ENDWHILE
    \IF{$x_{\module m}=0$ or $y_{\module m}=0$ and $\module m$ is not compact}
        \STATE Move $\module m$ (through the plane $x=-1$ or $y=-1$) to a dominated empty cell $p$\label{ln:move-moduleM}
    \ELSE
        \STATE Let $p$ be the closest dominated empty position to $\module m$ where $x_p=x_{\module m}-1$ (resp., $y_p=y_{\module m}-1$), and $\module q$ be the module at $p+(0,-1,0)$ (resp., $p+(-1,0,0)$)\label{ln:closest-p}
        \STATE{Move $\module q$ to $p$ and let $q$ be the cell left empty originally occupied by $\module q$}\label{ln:moveQ}%
    \FOR{ $i \in \{1, \dots, (z_\module m - z_q)\}$} %
        \STATE{ Move the module at $q+(0,0,i)$ to $q+(0,0,i-1)$ via slide}\label{ln:moveDown} %
    \ENDFOR %
    \STATE Slide the module at $(x_p,y_p,z_{\module{m}})$ to $\module m+(-1,-1,0)$ and slide $\module m$ to $(x_p,y_p,z_{\module{m}})$\label{ln:finalMove} 
    \ENDIF
    \RETURN
  \end{algorithmic}
\end{algorithm}

\begin{figure}[h!t]
    \centering
    \includegraphics[width=0.5\linewidth]{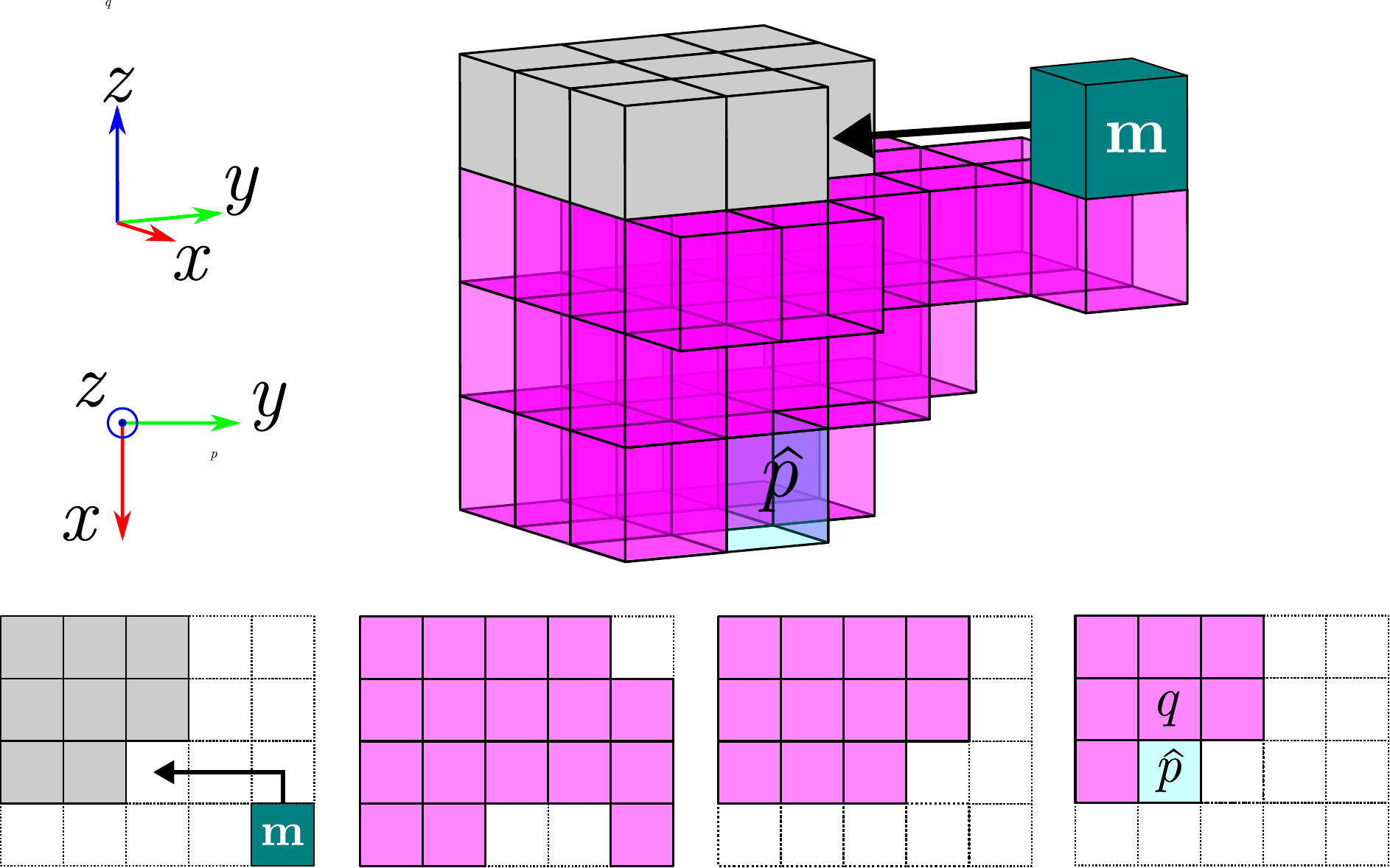}
    \caption{$\module m$ moves to a fixable position. Pink cells are occupied, but transparent for clarity. Grey modules are modules in $\square(\module m) \setminus \{\module m\}$, note each is quasi-compact.}
    \label{fig:fix-xray}
\end{figure}
}

After completing our post-order traversal, 
there is at least one straggler $\module m$.
Since $\module m$ is not quasi-compact, there is at least one empty position in the cuboid region dominated by $\module m$ and not directly below it.
In $\Fix(\module m)$ we search for such an empty position $p$ and effectively move $\module m$ to $p$, either directly or through a chain of moves that fill $p$ and leave $m$ empty.

. 
For a given cell $p=(x_p,y_p,z_p)$, we denote by $\square(p)$ the set of cells $(x_p-i,y_p-j,z_p)$ for all $0\le i\le x_p$ and $0\le j\le y_p$, i.e., the cells intersecting the minimum axis aligned horizontal square containing $(0,0,z_p)$ and $p$.
We define that a straggler is \emph{fixable} if every \textit{module} in $\square (\module m)\setminus\{\module m\}$ in the same cluster of $\module m$ is quasi-compact. 
We first try moving $\module m$ closer to the origin in $\square(\module m)$ (\cref{fig:fix-seq-a}).
If we can't, either (i) $\module m$ moves to a $z$-layer below it or becomes quasi-compact (and we are done); (ii) $\module m$ is on the edge of the bounding box (and we can move through the plane $x=-1$ or $y=-1$ to an empty position); or (iii) $\module m$ is blocked by other modules in the same layer that are quasi-compact (as in \cref{fig:fix-seq}).
Since $\module m$ is fixable and non-quasi-compact, there is an empty position $p$ dominated by $\module m$ in column $(x_{\module m} -1, y_{\module m},.)$ or $(x_{\module m}, y_{\module m}-1,.)$, adjacent to a module $\module q$ in column $(x_{\module m} -1, y_{\module m}-1,.)$ (\cref{fig:fix-seq-a}).
A sequence of moves fills $p$ and frees $(x_{\module m} -1, y_{\module m}-1,z_{\module m})$ (\cref{fig:fix-seq-b}).
We can then fill that empty position while freeing $\module m$'s cell (\cref{fig:fix-seq-c}).

\begin{figure}[h!t]
    \centering    
    \newcommand\scl{.15}
    \begin{subfigure}[b]{0.5\linewidth}
        \centering        \includegraphics[scale=\scl]{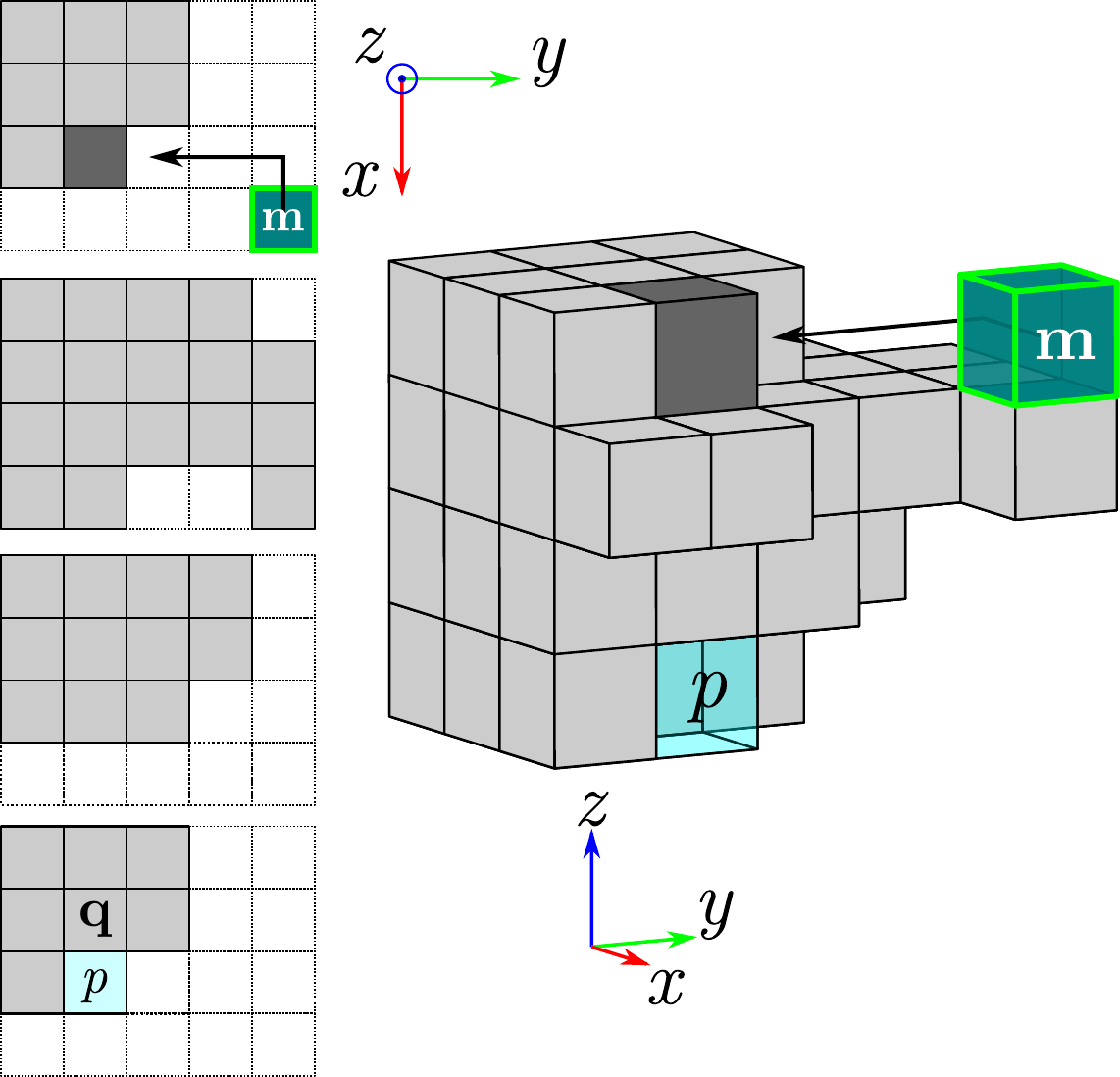}
        \caption{}
        \label{fig:fix-seq-a}
    \end{subfigure}%
    \hfill    
    \begin{subfigure}[b]{0.5\linewidth}
        \centering        \includegraphics[scale=\scl]{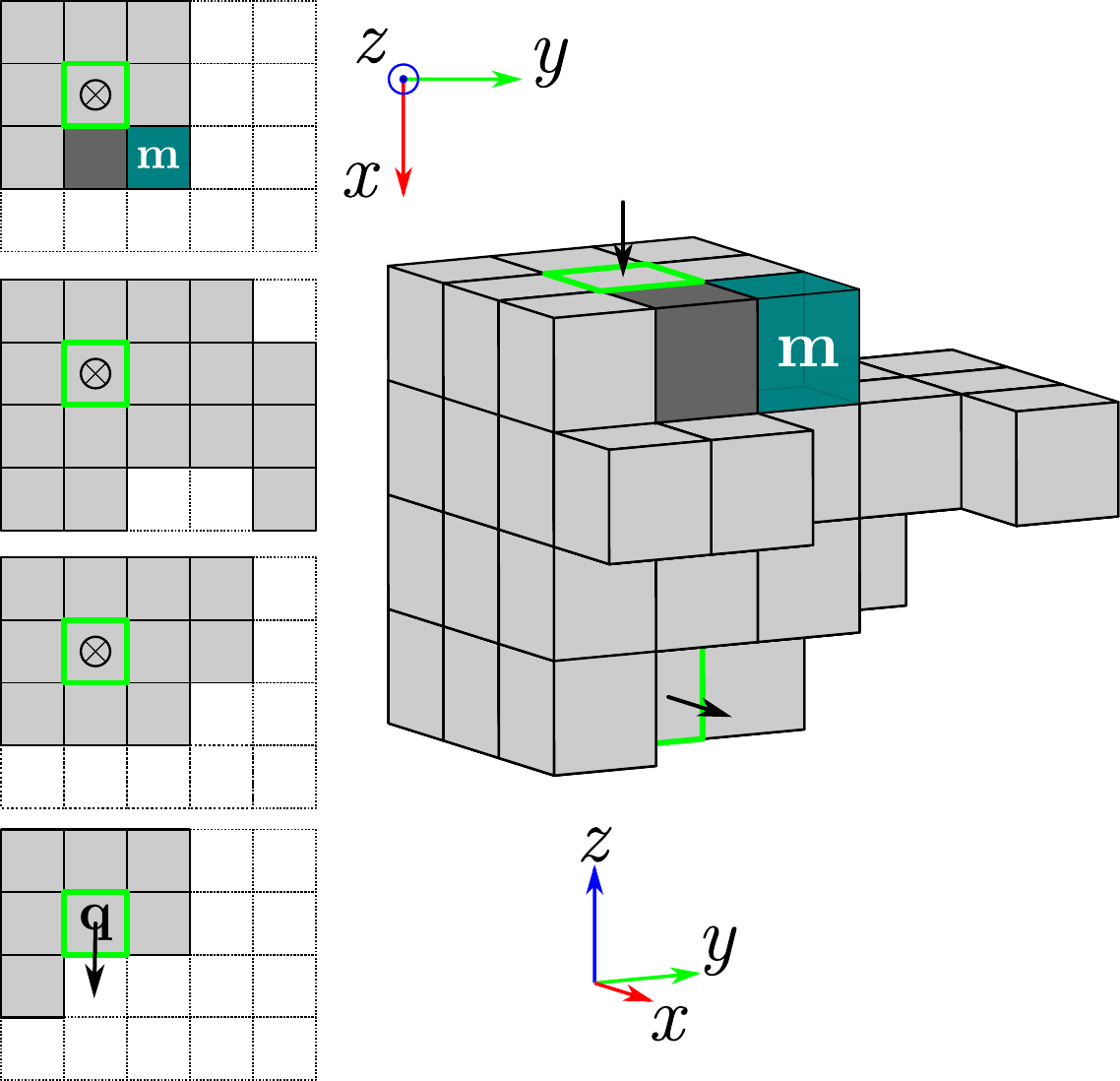}
        \caption{}
        \label{fig:fix-seq-b}
    \end{subfigure}%
    \par    
    \begin{subfigure}[b]{0.5\linewidth}
        \centering        \includegraphics[scale=\scl]{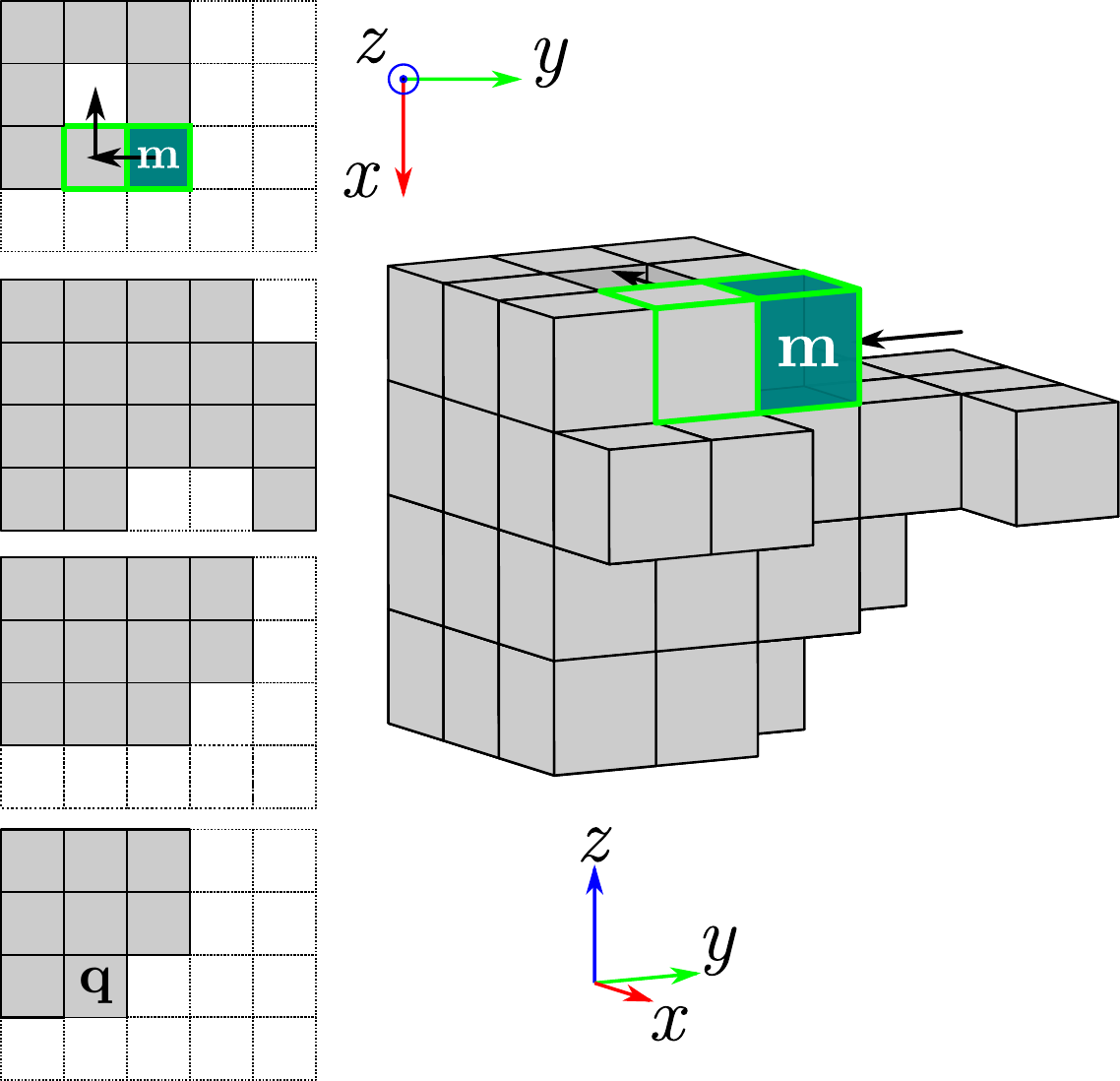}
        \caption{}
        \label{fig:fix-seq-c}
    \end{subfigure}%
    \hfill    
    \begin{subfigure}[b]{0.5\linewidth}
        \centering        \includegraphics[scale=\scl,page=4]{Figures/fix-seq.pdf}
        \caption{}
        \label{fig:fix-seq-d}
    \end{subfigure}%
        \caption{Illustration of the moves in $\Fix$ (a) line \ref{ln:fix-greedy}, (b) lines \ref{ln:moveQ}--\ref{ln:moveDown}, and (c) line \ref{ln:finalMove}. (d) After $\Fix$, $\module m$ becomes quasi-compact.
        \vspace{-.75em}}
        \label{fig:fix-seq}
    \end{figure}

\begin{lemma}
    If $\module m$ is fixable, \hyperref[alg:fix]{$\Fix(\module m)$} does not disconnect the configuration $C$, and performs at most $\|\module m-p\|_1+2$ moves.
    After $\Fix(\module m)$, if $\module m$ remained in the same $z$-layer, $\module m$ is now quasi-compact.
    \label[lemma]{lem:fix-connected}
\end{lemma}
\begin{proof}
Since $\module m$ is a straggler, if it is the only module moved by $\Fix(\module m)$, it is clear that connectivity is preserved. 
Lines \ref{ln:fix-first-line}--\ref{ln:move-moduleM} move $\module m$ monotonically towards the origin, except potentially for the move that brings $\module m$ outside of the bounding box (line~\ref{ln:move-moduleM}). 
The bound on the moves holds in this case, and if $\module m$ moved down (in $z$) the claim holds.

It remains to prove the claim for case (iii), i.e., lines \ref{ln:closest-p}--\ref{ln:finalMove}.
By the fact that $\module m$ is fixable, when we reach line~\ref{ln:closest-p}, cells $\module m+(-1,0,0)$, $\module m+(0,-1,0)$ and $\module m+(-1,-1,0)$ must be filled by quasi-compact modules. 
This implies that all cells with nonnegative coordinates strictly smaller than $\module m$'s are full.
Since $\module m$ is not yet quasi-compact, there must exist at least one empty position dominated by $\module m$ and not directly below it. 
Indeed, all such empty cells must be in columns $(x_{\module m} -1, y_{\module m},.)$ or $(x_{\module m}, y_{\module m} -1,.)$, otherwise $\module m+(-1,0,0)$ or $\module m+(0,-1,0)$ would not be quasi-compact.
Then, $p$ exists and $\module q$ is a compact nonarticulate module.
Thus, line~\ref{ln:moveQ} maintains connectivity.
Note that by the choice of $p$, every cell edge-adjacent in the same $z$-layer to a module moved by line~\ref{ln:moveDown} is full except for the cell in the column $x_\module m$.
Thus, the neighborhood of the moving module remains connected and line~\ref{ln:moveDown} maintains connectivity. 
When we reach line~\ref{ln:finalMove}, the cell at $(x_p,y_p,z_{\module{m}}-1)$ is full (by the choice of $p$).
The fact that $\module m$ is a straggler implies that every edge-adjacent cell above  $\module m$ have been already processed by $\Fix$, so they are quasi-compact and thus do not require $\module m$ or $(x_p,y_p,z_{\module{m}})$ for connectivity.
Then, both moves in line~\ref{ln:finalMove} maintain connectivity.
Except for line~\ref{ln:moveQ}  and the move from $(x_p,y_p,z_{\module{m}})$ to $\module m+(-1,-1,0)$ (line~\ref{ln:finalMove}), all the moves are monotonic in the direction $p-\module m$. Thus the total number of moves is as claimed.
\vspace{-1em}
\end{proof}

\subsection{\textsc{Compactify} for planar configurations}
\label{sec:2d}
With a few small adjustments, our algorithm can be applied to planar configurations (when all modules lie in the $z=0$ plane). 
Note that we allow using planes $z=1$ and $z=-1$ to move musketeers and stragglers.
The clusters are now defined by the $y$-slice graph (maximal components induced by a fixed $y$ coordinate). Lowering modules follows the same principle, except that we lower the $y$-coordinates of the modules (cases are simpler, as modules have degree 2 or less within the cluster).
$\Fix$ also becomes substantially simpler as stragglers can simply move to the closest dominated empty cell through the $z=-1$ plane.

\both{\section{Analysis of \textsc{Compactify}}}

\both{\begin{lemma}
    \hyperref[alg:Comp]{\textsc{Compactify}$(C)$} transforms a configuration $C$ into a compact configuration maintaining connectivity.
\end{lemma}}
\later{\begin{proof}
    The moves in line~\ref{ln:lower}  maintain connectivity due to Lemma~\ref{lem:melt}.
    By the choice of $\module{m}$ in line~\ref{ln:straggler-order} , the first time that $\Fix$ is run, $\module{m}$ is fixable.
    By lemma~\ref{lem:fix-connected} and the choice of $\module{m}$ in line~\ref{ln:straggler-order}, $\module{m}$ is fixable in every subsequent call of $\Fix$. 
    Thus, Lemma~\ref{lem:fix-connected} guarantees connectivity in line~\ref{ln:cmptx-fixCall}.
    When we reach line~\ref{ln:compacting-quasi}, $C$ is quasi-compact.
    Then, every module is either compact or it must be in the outer boundary of $C$.
    Thus, for a noncompact module $\module m$ in line~\ref{ln:slide-m} $\module{m}+(1,0,0)$ and $\module{m}+(0,1,0)$  are empty and its other neighbors are quasi-compact and do not need $\module m$ for connectivity. 
\end{proof}}


\both{\begin{lemma}
\label{lem:2comp}
    A sequence of $17\sum_{\module m\in C} \|\module m\|_1 + 60n + O(1)$ moves suffices to transform a configuration $C$ with $n$ modules into a compact one. 
\end{lemma}}

\later{
\begin{figure}[t]
    \centering
    \begin{subfigure}[b]{0.5\linewidth}
        \centering        
        \includegraphics[width=\linewidth,page=1]{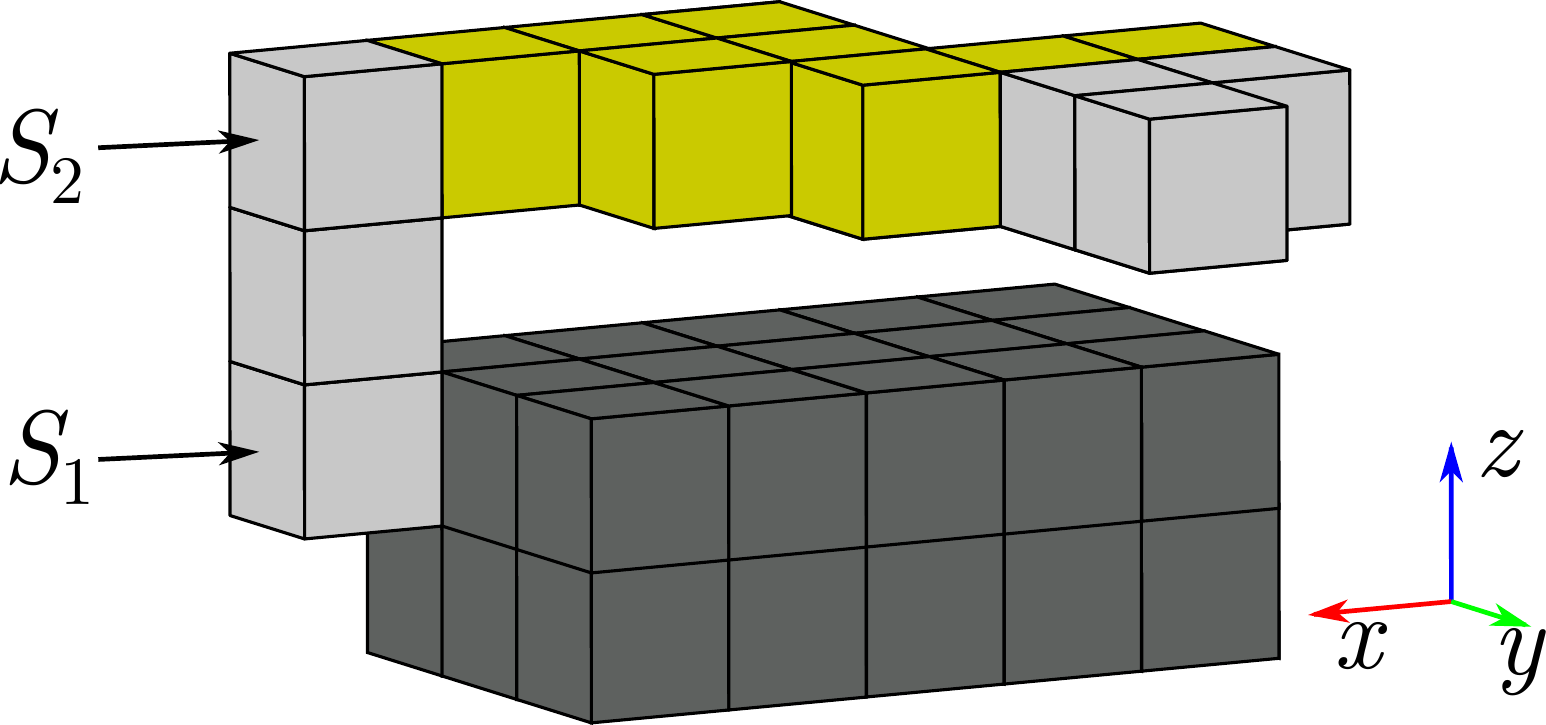}
        \caption{}
        \label{fig:cluster-charging-a}
    \end{subfigure}%
    \\
    \begin{subfigure}[b]{0.45\linewidth}
        \centering        
        \includegraphics[width=\linewidth,page=2]{Figures/cluster-charging.pdf}
        \caption{}
        \label{fig:cluster-charging-b}
    \end{subfigure}%
    \hfill
    \begin{subfigure}[b]{0.35\linewidth}
        \centering        
        \includegraphics[width=\linewidth,page=3]{Figures/cluster-charging.pdf}
        \caption{}
        \label{fig:cluster-charging-c}
    \end{subfigure}%
        \caption{(a) A configuration with quasi-compact modules shown in dark gray. Cluster $S_2$ has modules shown in yellow that become compact stragglers.
        (b--c) $S_2$ becomes $S$ after 2 iterations of the main while loop (line~\ref{ln:comp-while}) of \textsc{Compactify}. (c) The movement of one musketeer is shown by the red path.}
        \label{fig:cluster-charging}
    \end{figure}
}
    \begin{figure}
        \centering
        \includegraphics[scale=.75]{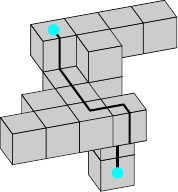}
        \caption{The path of an upwards traversal of musketeers from one cluster to another. Their start and end positions marked by blue circles.}
        \label{fig:upwards-traversal}
    \end{figure}
\later{
\begin{proof}

    By Lemma~\ref{lem:melt}, we can charge 17 moves to the decrease of the $z$-coordinate of the lowered modules.
    By Lemma~\ref{lem:melt} and \ref{lem:fix-connected},  the non-quasi-compact stragglers (in $S^*$) reduce their coordinates by at least one using a minimum of $16$ moves (13 from Lemma~\ref{lem:melt} and at least 3 from Lemma~\ref{lem:fix-connected}).
    However, stragglers that are already quasi-compact do not move, but the musketeers might move through them. 
    Refer to Fig.~\ref{fig:cluster-charging}.
    Let $S_2$ be the uppermost cluster shown in \cref{fig:cluster-charging-a} that becomes a cluster $S$ after 2 iterations of the main while loop of \textsc{Compactify}.
    Recall that, when processing $S$, $T$ is a minimum Steiner tree of the modules in $S^*$ (shown in light gray in \cref{fig:cluster-charging-b}).
    If $T$ is substantially larger than $S^*$, this would lead to more moves than we can afford if we simply count the decrease in coordinates of modules in $S^*$. (It takes $O(|T|)$ to process $S$, but the sum of coordinates only change by $O(|S*|)$.)
    Let $z_S$ be the $z$-coordinate of $S$.
    Recall that $T\setminus S^*$ contains only quasi-compact modules.
    Such modules are already in their final position and therefore they won't move, but they still ``hold potential'' (their sum of coordinates have not yet been charged to moves). 
    If such modules have not already been charged, then we can charge the movement of the musketeers through them to their sum of coordinates.
    Else, they have already been charged and thus were previously part of  another cluster $S_1$ previously in the slice $z_S$, into which another cluster $S_2$ merged after it was lowered by a previous iteration of \textsc{Compactify}, resulting in cluster $S$.
    (The dark gray modules that the musketeer walks on in \cref{fig:cluster-charging-c} were already quasi-compact in cluster $S_1$ shown in \cref{fig:cluster-charging-a})
    After that execution of line~\ref{ln:lower}, and subsequent executions of \textsc{Fix}, every module that remain in the processed slice is quasi-compact (becomes dark gray). 
    Because $S$ has dark gray modules that were already charged, all light gray modules ($S^*$) must come from a cluster previously located above $z_S$ that just got lowered by \textsc{Compactify}.
    Then, $S_2$ originally contained the modules in $S^*$, and not the modules in $T\setminus S^*$.    
    Refer to Figs.~\ref{fig:cluster-charging}(a--b).
    Since $S_2$ is connected, it must contain a path $p$ of modules connecting the different components of $S^*$ so that the modules in $p$ became semi-compact ($p$ must be a subset of the yellow modules that have spare potential). 
    By definition of $T$, $|p|\ge |T\setminus S^*|$.
    These modules store at least $z_S$ potential that can be used to pay for the movement of musketeers between the components of $S^*$ in this and all subsequent executions of line~\ref{ln:lower} involving such modules.
    Thus, we can charge at most  13 moves for each unit of $z$-coordinate of a module that becomes quasi-compact.
    Note that the planar version has strictly better constants for each unit of $y$-coordinate lowered.
    Overall, the bottleneck is 17 moves per lowered coordinate.

    By Lemma~\ref{lem:buildMusketeers} we can obtain our six musketeers in $6(4n) + O(1) = 24n + O(1)$ many moves. 
    We now only have to bound the number of moves conducted when moving musketeers between clusters (line~\ref{ln:muskToCluster} of \textsc{Compactify}).
    We claim this bound is $18n + O(1)$.
    For each module $\module m$ in a cluster $R$, we first bound the number of times the musketeers might visit the bottom face of $\module m$, then the side face of $\module m$, and finally the top faces of $\module m$.
    
    If the musketeers are sitting on $R$, and $R$ is not an extremal cluster, then the musketeers move to an extremal cluster $S$. 
    $S$ must be in a strictly greater $z$-layer than $R$, and further, the musketeers move to $S$ by walking only along clusters in $z$-layers strictly above $R$.
    Call this traversal from $R$ to $S$, an \textit{upwards traversal}, and correspondingly cluster visited was \textit{traversed upwards}.
    
    We argue a cluster may only be 
    traversed upwards once. 
    When a cluster is lowered it merges with every cluster in the layer below that had a module adjacent to it. 
    Therefore, once the musketeers traverse to a cluster $S$ above $R$, the musketeers will never reach a position in a $z$-layer below $R$, until $R$ has been lowered by line~\ref{ln:lower} of \textsc{Compactify}.  

    During an upwards traversal, in the worst case every musketeer walks on the bottom face of each module of a cluster exactly once, and walks on the side face of one module exactly once (see \cref{fig:upwards-traversal} for example).
    For each cluster, only a single side face is walked on and, if the musketeers are traversing the cluster upwards, at least one of the bottom faces of the cluster is not accessible (it connects to a cluster below it).
    Thus, among bottom and side faces, on average one face per module  is visited by the musketeers.    
    This amounts $6n + O(1)$ musketeer movements. 
    In the planar version of \textsc{Compactify} the musketeers can move in the positive $y$-direction using the bottom faces, which adds $6$ more moves per visited module.
    The total amounts $12n + O(1)$ musketeer movements.

    The musketeers may walk on the top face of a module at most two times.
     For any nonextremal cluster $R$, it has some modules with a neighbor above them. Consider a Steiner tree $T$, along the top faces of the modules of $R$, connecting these ``root'' modules. 
     When the musketeers are preforming upwards traversals. They walk along the top face of the modules of $R$.
     We charge musketeer movements to an Euler tour of $T$. 
     When a cluster above $R$, say $R^\uparrow$ is lowered and merged with $R$, the musketeers are placed on top of some cell of $R^\uparrow$ we can grow $T$ from $R$ to include the musketeers as a Steiner vertex, only using modules of $R^\uparrow$ (\cref{fig:musk-charging}). Therefore the musketeers walk along the top face of a module at most two times (each vertex of $T$ is visited at most two times). 
     The same is true for the planar version of \textsc{Compactify}. 
     Thus the bound on this type of musketeer moves is $24n$.

     Overall, the musketeers preform $24n + 12n + O(1) = 36n + O(1)$ moves. Combining this bound with the one for obtaining musketeer modules, we get a total of $24n + 36n + O(1) = 60n + O(1)$ moves.
\end{proof}
}
\later{
\begin{figure}[ht]
    \centering
    \includegraphics[width=0.75\linewidth]{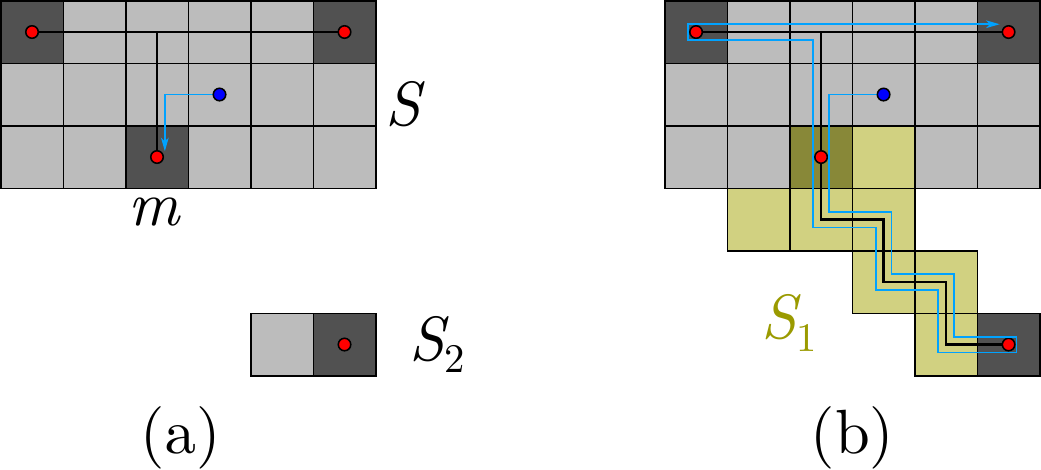}
    \caption{(a) A cluster $S$ and a Steiner tree $T'$ connecting modules adjacent to higher clusters. Another cluster $S_2$ is also in the same slice.
    The initial position of a musketeer is shown in blue.
    (b) During $\textsc{Compactify}$, $S$ merges with $S_1$ and $S_2$. The Steiner tree is updated using only modules not originally in $S$. The musketeer movements can be charged to a traversal in such a tree.}
    \label{fig:musk-charging}
\end{figure}
}

\both{
\begin{lemma}
\label{lem:runtime}
	$\textsc{Compactify}$ has runtime that is within an amortized $O(1)$ factor of the moves performed.
\end{lemma}}

\begin{proof}
Whenever a musketeer moves, it is either following a constant size schedule of moves such as \cref{fig:lower-straight}, or is traveling along a shortest path from one cluster to another. The former case clearly only requires $O(1)$ computation, and in the latter, computing the path requires time proportional to the length of the path, hence amortized $O(1)$ computation per movement.
The only section of this algorithm where an $O(1)$ computation cost is not simple is $\Fix$, when we must locate an appropriate empty position $p$. 
We show if each module stores a constant amount of information, these positions can be located efficiently.
Assume that initially each module $\module m$ has knowledge of the following three properties: 1) Is $\module m$ compact? 2) Is $\module m$ quasi-compact? 3)	 Is every position of $\square(\module m)$ filled? 

Note, each of these are transitive, and, to answer these questions, $\module m$ need only query its neighbors with smaller sum of coordinates. 
Consequently, if $\module m$ moves, recomputing if $\square(\module m)$ is full takes only $O(1)$ time.
If $\module m$ becomes (quasi-)compact it may be the case that many modules also become (quasi-)compact. 
While this may affect many more modules than the number of moves performed, a module can become (quasi-)compact only once. 
Therefore, in aggregate these updates will only take $O(n)$ time, which can then be amortized over the number of moves to $O(1)$ time.

If $x_\module m = 0$ or $y_\module m = 0$ (case(ii)), $\Fix$ attempts to move it to an empty dominated position $p$. If $\module m$ is not compact, then $p$ exists. Therefore we can determine if $p$ exists in $O(1)$ time. If $\module m$ is not compact, we can locate $p$ with $||\module m - p||_1$ queries: 
Checking positions directly below $\module m$ (from $\module m$) we will either find an empty position (and we are done) or a compact module by the fact that $\module m$ is fixable. 
We have found the $z$ coordinate of $p$ and we can find it by looking in the $-x$ or $-y$ direction (note that $\module m$ and $p$ lie in the same plane $x_\module m = 0$ or $y_\module m = 0$).

In case (iii), since we know $p$ is either in column $(x_{\module m} -1, y_{\module m},.)$ or $(x_{\module m}, y_{\module m} -1,.)$, we can find it with $||\module m - p||_1$ queries with a linear search alternating between the two columns.
\end{proof}

\later{\begin{proof}
Although the posed algorithm requires some complex operations, a relatively simple implementation requires $O(1)$ amortized computation per module movement.
Note, for most movements, this $O(1)$ bound is relatively trivial. 
A naive implementation for computing musketeer movements leads to an $O(1)$ computational overhead. 
By Lemma~\ref{lem:melt} lowering a module during $\textsc{Compactify}$, takes at most 17 movements. When musketeers reach a high degree vertex and must move to a leaf, we preform a DFS and then the musketeers move a number of modules proportional the length of this search. 
After we have lowered a slice and the musketeer modules must seek out a new extremal slice, again, a DFS is preformed and the musketeers preform a number of movements proportional to this search.

Truly, the only section of this algorithm where an $O(1)$ computation cost is not immediate is $\Fix$, when we must locate an appropriate empty position for $\Tunnel$. 
However, if we assume each module stores a constant amount of information, these positions can be located efficiently.
Assume that initially each module $\module m$ has knowledge of the following three properties: 1) Is $\module m$ compact? 2) Is $\module m$ quasi-compact? 3)	 Is every position of $\square(\module m)$ filled? 

Note, each of these are transitive, and, to answer these questions, $\module m$ need only query its neighbors with smaller sum of coordinates. For the first and second, if all such neighbors are compact or quasi-compact, then so is $\module m$. For the third property $\module m$ only has to communicate with its two neighbors in $\square(\module m)$, if both return true then $\square(\module m)$ is full as well. 

Consequently, if $\module m$ moves, recomputing if $\square(\module m)$ is full takes only $O(1)$ time.
However, if after a movement, if $\module m$ becomes compact it may be the case that many modules also become compact. 
When $\module m$ becomes compact, it must be the case that $\module m - (0,0,1)$, $\module m - (0,1,0)$, and $\module m - (1,0,0)$ are compact. The three other neighboring cells $\module m + (0,0,1)$, $\module m + (0,1,0)$, and $\module m + (1,0,0)$, may now be compact. If any are compact, say $\module m + (0,0,1)$, we check the same three relative positions $\module m + (0,0,2)$, $\module m + (0,1,1)$, and $\module m + (1,0,1)$ and continue this search. If they are empty or not compact, we stop checking neighbors. 
Although this search may potentially take much longer than the number of moves performed, a module can become compact only once. 
Every argument about the compactness of $\module m$ translates directly to the quasi-compactness of $\module m$.
Therefore, in aggregate these updates will only take $O(n)$ time, which can then be amortized over the number of moves to $O(1)$ time. 
Further
the transitivity of these properties means for a given configuration, precomputing these values can be accomplished by a simple BFS ($O(n)$ time).

Now, assuming each module has this knowledge, we can efficiently find our empty position $p$ for $\Fix$. 
During a given call of $\Fix(\module m)$.  
We first move $\module m$ monotonically to positions with smaller sum of coordinates. For each move a query to each neighboring cell with smaller coordinate sums is sufficent. So each of these moves only requires an $O(1)$ amount of runtime.

Then if $x_\module m = 0$ or $y_\module y = 0$ $\Fix$ attempts to move it to an empty dominated position $p$. If $\module m$ is not compact, then $p$ exists. Therefore we can determine if $p$ exists in $O(1)$ time. If $\module m$ is not compact, we can locate $p$ with $||\module m - p||_1$ queries. 
Without loss of generality, assume $x_\module m = 0$. 
We start by searching each position directly below $\module m$, $\module m - (0,0,i)$ for $i = 1$ to $i = z_\module m$. We stop this search when we find one of three things happen.

\begin{description}
    \item[$\module m - (0,0,i)$ is empty:] In this case we have found $p$.
    
    \item[$\module m - (0,0,i)$ is compact:] Then every position dominated by $\module m - (0,0,i)$ is full. However, since $\module m$ is not compact, the compactness of $\module m - (0,0,i)$ implies an empty position $p$ dominated by $\module m$ is present in the $z$ layer above $\module m - (0,0,i)$ (i.e. with $z$-coordinate $z_\module m - (i - 1)$). We can find $p$ by querying positions $\module m - (0,j,i - 1)$ for $j = 1$ to $j = y_\module m$, until one is empty.
    
    \item[$i = z_\module m$ and $\module m - (0,0,z_\module m)$ is not compact:] If $\module m - (0,0,z_\module m)$ is not compact, there must an empty cell with smaller $y$ coordiante. We preform a search similar to the last case to find it, and this becomes our $p$.
\end{description}
In each of these cases, we query a number of cells exactly equal to the manhattan distance between $\module m$ and $p$.

Finally, if $x_\module m$ and $y_\module m$ are both $>0$, $\Fix$ again searches for an empty position $p$ dominated by $\module m$. As noted in the proof of Lemma~\ref{lem:fix-connected}, $p$ must be directly below a neighbor of $\module m$. 
Therefore, it can be found in $z_\module m - z_p$ time, by querying positions directly below $\module m$. The number of moves that $\Fix$ then performs is precisely $z_\module m - z_p + 3$. Hence, amortizing the $z_\module m - z_p$ query time by the number of moves, we get an $O(1)$ amortized runtime. 
If $\module m$ is quasi-compact, then $p$ does not exist. 
Then, the exit condition of the while loop in line \ref{ln:comp-while} can be checked in $O(1)$ time.
\end{proof}}

\section{Reconfiguring Between Compact Configurations}
\label{sec:comp2comp}
\iffull
The final step for universal reconfiguration is a procedure that reconfigures between compact configurations.
\fi

\begin{lemma}
    \label{lem:comp2comp}
    Given two compact configurations $C_1$ and $C_2$ with $n$ modules where all modules lie in the positive $xyz$ orthant, we can reconfigure one into the other by using at most $\sum_{\module m\in C_1} \|\module m\|_1 + \sum_{\module m\in C_2} \|\module m\|_1 $ sliding moves. 
\end{lemma}

\begin{proof}
    We use induction on $|C_1\setminus C_2|$.
    The base case is when $|C_1\setminus C_2|=0$ and thus $C_1=C_2$ and we are done.
    Otherwise, we claim that there is a module $\module m_1 \in C_1\setminus C_2$ whose cell is not dominated by any other module in $ C_1\cup C_2$. 
    We prove this by contradiction: assume that every module if $C_1\setminus C_2$ is dominated by some module in $ C_1\cup C_2$. Because domination establishes a partial order, there must exist a module $\module m \in C_1\setminus C_2$ that us dominated by a module in $\module m'\in C_1\cap C_2$. That is, we have found a module $\module m$ is only present in $C_1$ and is dominated by $\module m' \in C_1\cap C_2$ that is present in both $C_1$ and $C_2$. However, this is a contradiction to the fact that $C_2$ is compact.
    The interesting property of $\module m_1$ is that  $C_1\setminus\{\module m_1\}$ is compact.

    Using a symmetric argument, we show that there exists a cell $\module m_2\in C_2\setminus C_1$ so that every dominated cell is occupied in $C_1\cap C_2$.
    For contradiction, assume no such module exists. Let $\module m$ be the module in $C_2\setminus C_1$ with smallest $\|\module m\|_1$.
    Since $C_2$ is compact, all cells that $\module m$ dominates are occupied in $C_2$.
    Let $m'\in C_2\setminus C_1$ be a module dominated by $\module m$, which must exist by assumption. 
    Then $\|\module m'\|_1 < \|\module m\|_1$, contradicting the choice of $\module m$.
    We conclude that $C_1\cup \{\module m_2\}$ is compact. Furthermore, the two claims combined imply that $C_1'=C_1\cup \{\module m_2\}\setminus \{\module m_1\}$ is compact.
    Then $|C_1\setminus C_2| = |C_1'\setminus C_2|+2$ as desired.

    To complete the proof we must show that the reconfiguration can be done in the claimed number of moves. Note that $\module m_1$ and $\module m_2$ must both be in the outer boundary of $C_1$ and $C_1'$ respectively. 
    By definition of compact, the boundary of $C_1\setminus\{\module m_1\}$ that is not on the $x=0, y=0$ or $z=0$ planes is an $x$-, $y$-, and $z$-monotone surface.
    Thus, the shortest path from $\module m_1$ to $\module m_2$ on this surface is also monotone and, hence its length is upper bounded by $\|\module m_1\|_1+\|\module m_2\|_1$.
\end{proof}

\section{Conclusion}\label{sec:conclusion}
Algorithmic bounds have been derived from the number of modules~\cite{DP,abel2008universal}, to dimensions of bounding boxes~\cite{a.akitaya_et_al:LIPIcs.SWAT.2022.4,aakks-quadratic04} and eventually the sum of coordinates (this paper and ~\cite{inPlaceCompact24}).
We note that, none are optimal if the starting and target configurations are far from compact but very similar. Is there a better parameter to bound the length of reconfiguration between configurations? Such a bound is also important for approximation algorithms, of which there are none. Note that in general, shortest reconfiguration is known to be NP-complete~\cite{a.akitaya_et_al:LIPIcs.SWAT.2022.4}. 

Another avenue of research would be parallelization where we look into the {\em makespan} (i.e., time required to execute all moves when we allow robots to move independently). The constants in our algorithm would directly improve roughly by a factor of 6 if musketeers are allowed to move simultaneously. Makespan would be further improved if freed stragglers acted as another set of musketeers working in parallel.

\iffull
Recent research on practical prototypes of modular robots (such as the {ARMADAS} model used by NASA~\cite{realisticSwat24}) focus on the \emph{accessibility} of the moves in a reconfiguration.
Here, every move must occur on the outer boundary of the configuration. In other models, module moves can sometimes be executed only by a few active agents~\cite{agents25} (all other modules are passive and can only remain in place).
In our algorithm, every move is accessible from the outer face and happens close to our active agents (i.e., musketeers). The only exception is during the execution of line \ref{ln:moveDown} of $\Fix$. We believe that our algorithm can be adapted to fit either model while remaining in-place and having similar bounds.
\fi


\renewcommand{\em}{\itshape}
\renewcommand{\emph}[1]{\textit{#1}}


\small
\bibliographystyle{abbrv}
\bibliography{cubes}

\begin{thebibliography}{10}

\bibitem{aakks-quadratic04}
Z.~Abel, H.~A. Akitaya, S.~D. Kominers, M.~Korman, and F.~Stock.
\newblock A universal in-place reconfiguration algorithm for sliding
  cube-shaped robots in a quadratic number of moves.
\newblock In {\em SoCG}, pages 1:1--1:14, 2024.

\bibitem{abel2008universal}
Z.~Abel and S.~D. Kominers.
\newblock Universal reconfiguration of (hyper-)cubic robots.
\newblock {\em arXiv preprint}, 2008.

\bibitem{akitaya2021universal}
H.~A. Akitaya, E.~M. Arkin, M.~Damian, E.~D. Demaine, V.~Dujmovi{\'c},
  R.~Flatland, M.~Korman, B.~Palop, I.~Parada, A.~v.~R. Renssen, and
  V.~Sacristán.
\newblock Universal reconfiguration of facet-connected modular robots by
  pivots: The {$O(1)$ Musketeers}.
\newblock {\em Algorithmica}, 83:1316--1351, 2021.

\bibitem{a.akitaya_et_al:LIPIcs.SoCG.2021.10}
H.~A. Akitaya, E.~D. Demaine, A.~Gonczi, D.~H. Hendrickson, A.~Hesterberg,
  M.~Korman, O.~Korten, J.~Lynch, I.~Parada, and V.~Sacrist\'{a}n.
\newblock Characterizing universal reconfigurability of modular pivoting
  robots.
\newblock In {\em SoCG}, volume 189, pages 10:1--10:20, 2021.

\bibitem{a.akitaya_et_al:LIPIcs.SWAT.2022.4}
H.~A. Akitaya, E.~D. Demaine, M.~Korman, I.~Kostitsyna, I.~Parada, W.~Sonke,
  B.~Speckmann, R.~Uehara, and J.~Wulms.
\newblock Compacting squares: Input-sensitive in-place reconfiguration of
  sliding squares.
\newblock In {\em SWAT}, pages 4:1--4:19, 2022.

\bibitem{aloupis2009linear}
G.~Aloupis, S.~Collette, M.~Damian, E.~D. Demaine, R.~Flatland, S.~Langerman,
  J.~O'Rourke, S.~Ramaswami, V.~Sacrist{\'a}n, and S.~Wuhrer.
\newblock Linear reconfiguration of cube-style modular robots.
\newblock {\em Computational Geometry}, 42(6-7):652--663, 2009.

\bibitem{aloupis2008reconfiguration}
G.~Aloupis, S.~Collette, E.~D. Demaine, S.~Langerman, V.~Sacrist{\'a}n, and
  S.~Wuhrer.
\newblock Reconfiguration of cube-style modular robots using $o (\log n)$
  parallel moves.
\newblock In {\em ISAAC}, pages 342--353. Springer, 2008.

\bibitem{agents25}
A.~T. Becker, S.~P. Fekete, J.~Friemel, R.~Kosfeld, P.~Kramer, H.~Kube,
  C.~Rieck, C.~Scheffer, and A.~Schmidt.
\newblock Moving matter: Efficient reconfiguration of tile arrangements by a
  single active robot.
\newblock In {\em to appear in CCCG}, 2025.

\bibitem{DP}
A.~Dumitrescu and J.~Pach.
\newblock Pushing squares around.
\newblock {\em Graphs and Combinatorics}, 22(1):37--50, 2006.

\bibitem{dumitrescu2004motion}
A.~Dumitrescu, I.~Suzuki, and M.~Yamashita.
\newblock Motion planning for metamorphic systems: Feasibility, decidability,
  and distributed reconfiguration.
\newblock {\em IEEE Transactions on Robotics and Automation}, 20(3):409--418,
  2004.

\bibitem{fitch2003reconfiguration}
R.~Fitch, Z.~Butler, and D.~Rus.
\newblock Reconfiguration planning for heterogeneous self-reconfiguring robots.
\newblock In {\em Proceedings of the 2003 IEEE/RSJ International Conference on
  Intelligent Robots and Systems (IROS 2003)}, volume~3, pages 2460--2467.
  IEEE, 2003.

\bibitem{inPlaceCompact24}
I.~Kostitsyna, T.~Ophelders, I.~Parada, T.~Peters, W.~Sonke, and B.~Speckmann.
\newblock Optimal in-place compaction of sliding cubes.
\newblock In {\em {SWAT}}, pages 31:1--31:14, 2024.

\bibitem{sung2015reconfiguration}
C.~Sung, J.~Bern, J.~Romanishin, and D.~Rus.
\newblock Reconfiguration planning for pivoting cube modular robots.
\newblock In {\em Proceedings of the 2015 IEEE International Conference on
  Robotics and Automation (ICRA)}, pages 1933--1940. IEEE, 2015.

\bibitem{realisticSwat24}
M.~S.~R. Team, J.~Brunner, K.~C. Cheung, E.~D. Demaine, J.~Diomidova, C.~Gregg,
  D.~H. Hendrickson, and I.~Kostitsyna.
\newblock Reconfiguration algorithms for cubic modular robots with realistic
  movement constraints.
\newblock In {\em {SWAT}}, pages 34:1--34:18, 2024.

\end{thebibliography}

\magicappendix
\end{document}